\documentclass[11pt]{scrartcl}

\usepackage{url}
\usepackage[hidelinks]{hyperref}
\usepackage[utf8]{inputenc}
\usepackage[small]{caption}
\usepackage{graphicx}
\usepackage{amsmath}
\usepackage{booktabs}
\usepackage{algorithm,algpseudocode}
\usepackage{amsthm}
\usepackage{amssymb}

\usepackage{cleveref}
\usepackage{color}
\usepackage{xspace}
\usepackage{tikz}
\usepackage{tabularx}
\usepackage{pbox}
\usepackage{framed}

\usetikzlibrary{positioning}

\DeclareOldFontCommand{\bf}{\normalfont\bfseries}{\mathbf}

\newtheorem{theorem}{Theorem}[section]
\newtheorem{corollary}[theorem]{Corollary}
\newtheorem{proposition}[theorem]{Proposition}
\newtheorem{lemma}[theorem]{Lemma}

\newtheorem{open}[theorem]{Open question}

\theoremstyle{definition}
\newtheorem{definition}[theorem]{Definition}

\newtheorem{remark}[theorem]{Remark}

\algnotext{EndFor}
\algnotext{EndIf}
\algnotext{EndWhile}
\algnotext{EndProcedure}

\usepackage{natbib}

\usepackage{enumitem}
\newcommand{\mms}{\text{MMS}}

\newcommand{\eval}{\textsc{Eval}}
\newcommand{\cut}{\textsc{Cut}}
\newcommand{\cutr}{\textsc{CutRight}}

\newcommand{\allocation}{\mathbf{A}}
\newcommand{\partition}{\mathbf{P}}
\newcommand{\partitions}{\Pi}

\newcommand{\findsum}[1]{\textsc{FindSum1}($#1$)}

\allowdisplaybreaks

\definecolor{ForestGreen}{rgb}{.13,.54,.13}

\begin{document}

\title{Mind the Gap: \\Cake Cutting With Separation}

\author{
Edith Elkind\\University of Oxford
\and
Erel Segal-Halevi\\Ariel University
\and
Warut Suksompong\\National University of Singapore
}

\date{\vspace{-3ex}}

\maketitle

\begin{abstract}
We study the problem of fairly allocating a divisible resource, also known as cake cutting, with an additional requirement that the shares that different agents receive should be sufficiently separated from one another.
This captures, for example, constraints arising from social distancing guidelines.
While it is sometimes impossible to allocate a proportional share to every agent under the separation requirement, we show that the well-known criterion of maximin share fairness can always be attained.
We then provide algorithmic analysis of maximin share fairness in this setting---for instance, the maximin share of an agent cannot be computed exactly by any finite algorithm, but can be approximated with an arbitrarily small error.
In addition, we consider the division of a pie (i.e., a circular cake) and show that an ordinal relaxation of maximin share fairness can be achieved.
We also prove that an envy-free or equitable allocation that allocates the maximum amount of resource exists under separation.
\end{abstract}

\section{Introduction}

The end of the year is fast approaching, and members of a city council are busy planning the traditional New Year's fair on their city's main street.
As usual, a major part of their work is to divide the space on the street among interested vendors.
Each vendor naturally has a preference over potential locations, possibly depending on the proximity to certain attractions or the estimated number of customers visiting that space.
Additionally, this year is different from previous years due to the social distancing guidelines issued by the government---vendors are required to be placed at least two meters apart.
How should the city council allot the space so that all vendors feel fairly treated and at the same time everyone stays safe and sound under the new guidelines?

The problem of fairly allocating a heterogeneous divisible good among a set of agents in a fair manner has a long history and is commonly known as \emph{cake cutting} \citep{BramsTa96,RobertsonWe98,Procaccia16}.
A typical fairness criterion in cake cutting is \emph{proportionality}, which means that each agent should receive her proportionally fair share, i.e., $1/n$ of the agent's value for the whole cake, where $n$ denotes the total number of agents.
For any set of agents with arbitrary valuations, a proportional allocation in which each agent receives a single connected piece is guaranteed to exist.
Better still, such an allocation can be found by a simple and efficient algorithm \citep{DubinsSp61}.

In this paper, we initiate the study of cake cutting with separation requirements.
Besides the social distancing example that we mentioned, our setting captures the task of allocating machine processing time, where we need time to erase data from the previous process before the next process can be started, as well as land division, where we want space between different plots in order to avoid cross-fertilization.
When separation is imposed, it is no longer the case that proportionality can always be satisfied---an extreme example is when all agents value only a common small piece of length less than the minimum gap required.
A similar failure of proportionality has notably been observed in the allocation of \emph{indivisible} items (without separation), and a solution that has been proposed and widely studied in that context is \emph{maximin share fairness} \citep{Budish11,KurokawaPrWa18}.
Maximin share fairness requires each agent to receive her ``maximin share'' (MMS), which is the best share that the agent can secure by dividing the items into $n$ bundles and getting the worst bundle.
In this work, we demonstrate that maximin share fairness is an appropriate substitute for proportionality in cake cutting with separation, and analyze this concept from a computational perspective.
This is one of the first uses of maximin share fairness in cake cutting (see Section~\ref{sec:related-work}).

\subsection{Our Results}

As is commonly done in cake cutting, we assume that the cake is represented by an interval, and each agent is to be allocated a single subinterval of the cake.
We further require the pieces of any two agents to be separated by distance at least $s$, where $s> 0$ is a given separation parameter.
For the sake of exposition, we follow the convention in most of the literature and assume that the agents have additive valuations over the cake.
However, as we discuss in Section~\ref{sec:conclusion}, some of our positive results hold even for agents with arbitrary monotonic valuations.

\renewcommand{\arraystretch}{1.2}
\begin{table*}[!ht]
\centering
    \begin{tabular}{| c | c | c |}
    \hline
    Task & Cake Cutting & Pie Cutting \\ \hline \hline
    Decide whether $\mms_i=r$ & Yes (Cor.~\ref{cor:mms-query-equal}) & No (Thm.~\ref{thm:pie-mms-greater}) \\ \hline
    Decide whether $\mms_i>r$ & Yes (Thm.~\ref{thm:mms-query-greater}) & No (Thm.~\ref{thm:pie-mms-greater}) \\ \hline
    Decide whether $\mms_i\geq r$ & Yes  (Thm.~\ref{thm:mms-query-atleast}) & No (Thm.~\ref{thm:pie-mms-greater-equal}) \\ \hline 
    Compute the maximin share of an agent  & No (Thm.~\ref{thm:mms-impossibility}) & No (Cor.~\ref{cor:pie-mms-compute}) \\ \hline 
    Approximate the maximin share up to $\varepsilon$ & Yes (Cor.~\ref{cor:mms-query-bound}) & Yes (Thm.~\ref{thm:pie-mms-approx}) \\ \hline 
    Compute a maximin partition of an agent & No (Cor.~\ref{cor:mms-no-partition}) & No (Cor.~\ref{cor:pie-mms-compute}) \\
    \hline 
    Approximate a maximin partition up to $\varepsilon$ & Yes (Cor.~\ref{cor:mms-query-bound}) & Yes (Thm.~\ref{thm:pie-mms-approx}) \\ \hline 
    \end{tabular}
    \vspace{5mm}
    \caption{Summary of the tasks that can and cannot be accomplished by finite algorithms in the Robertson--Webb model for cake cutting and pie cutting. All negative results hold even when the valuations of the agents are piecewise constant (but not given explicitly).}
    \label{table:summary}
\end{table*}

In Section~\ref{sec:cake}, we begin by proving that 
maximin share fairness can be guaranteed:
there always exists an allocation that gives every agent at least her maximin share. If the maximin share of each agent is known, 
such an allocation can be found by a simple algorithm similar to the aforementioned algorithm by \citet{DubinsSp61}.
Unfortunately, we show that \emph{no} finite algorithm can compute the maximin share of an agent exactly in the standard Robertson--Webb model---this impossibility result holds even when $n=2$ and the agents have piecewise constant valuations.
To establish this result, we prove that no finite number of queries can solve a basic function problem that we call \textsc{FindSum1}, which may be of independent interest.
Nevertheless, we design an algorithm based on binary search that approximates the maximin share up to an arbitrarily small error. This enables us to compute an allocation wherein each agent obtains an arbitrarily close approximation of her maximin share.
In addition, we present algorithms that decide whether the maximin share of an agent is greater than, less than, or equal to a given value, and show that if the agents have piecewise constant valuations that are given \emph{explicitly} as part of the input, then we can compute their exact maximin shares, and therefore an MMS-fair allocation, in polynomial time using linear programming.

In Section~\ref{sec:pie}, we consider the allocation of a ``pie'', which is a one-dimensional circular cake and serves to model, for example, 
the streets around a city square,
the shoreline of an island, or daily time slots for using a facility.
In contrast to cake cutting, maximin share fairness cannot necessarily be guaranteed in pie cutting, and even the commonly studied cardinal multiplicative approximation cannot be obtained.
Therefore, we focus instead on an \emph{ordinal} relaxation of the maximin share, which allows each agent to partition the pie into $k$ pieces for some parameter $k > n$.
We show that when $k = n+1$, the resulting fairness guarantee---called the \emph{$1$-out-of-$(n+1)$ maximin share}---can be satisfied.
We then investigate computational properties of maximin share fairness in pie cutting, and demonstrate several similarities and differences with cake cutting.
In particular, while we can still approximate the maximin share of an agent (albeit less efficiently than in cake cutting), deciding whether the maximin share is greater than, less than, or equal to a given value is no longer possible for any finite algorithm.
A summary of our results in Sections~\ref{sec:cake} and \ref{sec:pie} can be found in Table~\ref{table:summary}.

Finally, in Section~\ref{sec:EF-EQ}, we investigate two other important fairness notions: envy-freeness and equitability.
While these notions can be satisfied trivially by not allocating any of the cake or pie, we show that there always exist allocations fulfilling each of these criteria while at the same time allocating the maximum possible amount of resource subject to separation constraints.

\subsection{Related Work}
\label{sec:related-work}

Cake cutting has long been studied by mathematicians and economists, and more recently attracted substantial interest from computer scientists, as it suggests a plethora of computational challenges.
In particular, a long line of work in the artificial intelligence community in recent years has focused on cake cutting and its variants \citep{BalkanskiBrKu14,LiZhZh15,BranzeiCaKu16,AlijaniFaGh17,BeiChHu17,MenonLa17,ArunachaleswaranBaRa19,HosseiniIgSe20}. 

In order to ensure that no agent receives a collection of tiny pieces, it is often assumed that each agent must be allocated a connected piece of the cake \citep{DubinsSp61,Stromquist80,Stromquist08,Su99,BeiChHu12,CechlarovaPi12,CechlarovaDoPi13,AumannDo15,ArunachaleswaranBaKu19,GoldbergHoSu20}. 
Indeed, when we divide resources such as time or space, non-connected pieces (e.g., disconnected time intervals or land plots) may be hard to utilize, or even entirely useless.
Note that we impose the connectivity constraint not only on the allocation but also in the definition of the maximin share benchmark.
Similar conventions have been used in the context of indivisible items, where the items are vertices of an undirected graph and every agent must be allocated a connected subgraph \citep{BouveretCeEl17,IgarashiP19,LoncTr20,BiloCFIMPVZ22}.\footnote{\citet{BeiIgLu22} explored the relations between the constrained and unconstrained versions of the maximin share in that context.}

Most previous works on cake cutting (and pie cutting) did not explicitly consider the maximin share. This is because, with additive utilities, a proportional allocation is also an MMS-fair allocation, since each agent's maximin share is always at most $1/n$ of the agent's value for the entire cake (or pie). 
In particular, without separation constraints, classic algorithms for proportional cake cutting \citep{Steinhaus1948Problem,DubinsSp61,Even1984Note} attain maximin share fairness.
The maximin share only becomes interesting when a proportional allocation may not exist.

We are aware of two recent studies of the maximin share in cake cutting.
\citet{bogomolnaia2020guarantees} considered agents with general continuous valuations---not necessarily additive or even monotonic. 
They showed that the maximin share is not always attainable, but the \emph{minimax share} (i.e., the worst-case share of an agent when the items are partitioned into $n$ bundles and the agent gets the best bundle) can always be guaranteed.
\citet[Appendix B]{Segalhalevi20-2} showed that maximin share fairness can be attained when the cake is a collection of disconnected intervals and each agent should receive a connected piece. 
We are not aware of previous studies of maximin share fairness in pie cutting.

\citet{iyer2005multiagent}
presented negotiation protocols for fairly dividing a cake or a pie. Their protocols require each agent to submit a set of $n$ intervals (in cake division) or $n+1$ intervals (in pie division), and guarantee to each agent one of her intervals. 
While presented in a different framework, their protocols are similar in spirit to our algorithms for a cake (Theorem~\ref{thm:mms-algo}) and for a pie
(Theorem~\ref{thm:pie-mms-algo}).
However, they did not consider separation, and focused on solution concepts other than the maximin share (indeed, recall that maximin share fairness is trivial in the absence of separation).

After the publication of the conference version of our work \citep{ElkindSeSu21}, we extended our study of separation constraints to the division of two-dimensional resources such as land \citep{ElkindSeSu21-Land} and graphical resources such as road networks \citep{ElkindSeSu21-Graph}.
In particular, our guarantees for arbitrary graphs generalize Theorems~\ref{thm:mms-algo} and \ref{thm:pie-mms-algo} of the present paper, but the general algorithms and their analysis are much more involved. Apart from this, there is no overlap between these papers.

\section{Preliminaries}
In cake cutting, the cake is represented by the interval $[0,1]$.
The set of agents is denoted by $N=[n]$, where $[k]:=\{1,2,\dots,k\}$ for any positive integer $k$.
The preference of each agent $i\in N$ is represented by an integrable
\emph{density function} $f_i:[0,1]\to \mathbb{R}_{\ge 0}$, which captures how the agent values different parts of the cake.
A \emph{piece of cake} is a finite union of disjoint closed intervals of the cake; it is said to be \emph{connected} if it consists of a single interval.
Agent $i$'s value for a piece of cake $X$ is given by $v_i(X) := \int_{x\in X}f_i(x) dx$.
For $0\leq x\leq y\leq 1$, we simplify notation and write $v_i(x,y) = v_i([x,y])$.
As is standard in cake cutting, we assume that the density functions are normalized so that $v_i(0,1) = 1$ for all $i\in N$.
A valuation function is said to be \emph{piecewise constant} if it is represented by a piecewise constant density function.
An \emph{allocation} of the cake is denoted by an $n$-tuple $\allocation=(A_1,\dots,A_n)$, where each $A_i$ is a piece of cake, and 
for all $i, j\in N$ such that $i\neq j$ the set $A_i\cap A_j$
consists of finitely many points.
The piece $A_i$ is allocated to agent $i$.
We are interested in allocations that are \emph{connected}, that is, each $A_i$ is a connected piece.

Let $s\ge 0$ be a real parameter.
We seek connected allocations in which any two pieces are separated by length at least $s$; we call such allocations \emph{$s$-separated}.
The case $s=0$ corresponds to the classic setting (without separation), and when $s\geq \frac{1}{n-1}$ an $s$-separated allocation must have at least one piece of length $0$. Therefore, from now on we assume that $s\in (0,\frac{1}{n-1})$.
A set $\partition=\{P_1,\dots,P_n\}$ is an (\emph{$s$-separated}) \emph{partition} if the tuple $\allocation=(P_1,\dots,P_n)$ is an (\emph{$s$-separated}) allocation; intuitively, a partition is a collection of pieces, without a specification of which agent gets which piece.
Assume without loss of generality that in each partition the pieces $P_1,\dots,P_n$ are listed in increasing order, i.e., for each $1\le i<j\le n$ we have $a\le b$ for all $a\in P_i$, $b\in P_j$.
The \emph{min-value} of partition $\partition$ for agent $i$ is defined as $\min_{j\in [n]} v_i(P_j)$.
Let $\partitions_{n, s}(x, y)$ denote the set of all $s$-separated partitions of the interval $[x, y]$ among $n$ agents
(where $[x, y]\subseteq [0, 1]$).
Note that  an $s$-separated allocation or partition is incomplete, since some of the cake necessarily remains unallocated.
An \emph{instance} consists of the agents, cake, density functions, and a separation parameter.

A standard method for a cake-cutting algorithm to access agents' valuations is through queries. Specifically, we use the model of \citet{RobertsonWe98}, which supports two types of queries:
\begin{itemize}
\item $\eval_i(x,y)$: Asks agent $i$ to evaluate the interval $[x,y]$ and return the value $v_i(x,y)$.
\item $\cut_i(x,\alpha)$: Asks agent $i$ to return the leftmost point $y$ such that $v_i(x,y)=\alpha$, or to state that no such point exists.
\end{itemize}

We now define the main fairness criterion of our paper.

\begin{definition}
\label{def:MMS}
The \emph{maximin share} of agent $i\in N$ with respect to 
an interval $[x, y]\subseteq[0, 1]$ is defined as 
\begin{align*}
\mms_i^{n,s}(x, y) := \sup_{\partition\in \partitions_{n, s}(x, y)}\min_{j\in [n]}v_i(P_j).
\end{align*}
When $n$ and $s$ are clear from the context, we omit them
from the notation and write $\mms_i(x, y)$ instead of 
$\mms_i^{n,s}(x, y)$. Moreover, when $[x, y]=[0, 1]$, we further abbreviate $\mms_i(x, y)$
to $\mms_i$.
\end{definition}

Let $\partitions'_{n,s}(x, y)\subseteq \partitions_{n, s}(x, y)$ be the set of all $s$-separated partitions of $[x, y]$ such that every pair of consecutive pieces is separated by length exactly $s$.
We claim that the definition of the maximin share can be simplified by replacing $\partitions_{n, s}(x, y)$ with $\partitions'_{n, s}(x, y)$; intuitively, for every partition in which the distance between some pair of adjacent pieces is larger than~$s$, there is a partition with at least the same min-value
in which the distance between all pairs of adjacent pieces is exactly $s$.
We also claim that the supremum in the definition can be replaced with a maximum, i.e., a maximizing partition always exists.
\begin{proposition}
\label{prop:alt-def}
For every agent $i\in N$, it holds that 
\[\emph{MMS}^{n,s}_i(x, y) = \max_{\partition\in \partitions'_{n,s}(x, y)}\min_{j\in [n]}v_i(P_j).\]
\end{proposition}

\begin{proof}
Fix a partition $\partition\in \partitions_{n, s}(x, y)$. Suppose that some pair of consecutive pieces of~$\partition$ is separated by length more than $s$. We can then extend one of these pieces so that they are separated by length exactly $s$. By doing so for all pairs of consecutive pieces, we obtain another partition $\partition'\in\partitions'_{n, s}(x, y)$ such that $P_j\subseteq P'_j$ for all $j\in [n]$. It follows that $\min_{j\in [n]}v_i(P'_j)\geq \min_{j\in [n]}v_i(P_j)$, so in Definition~\ref{def:MMS} we may replace $\partitions_{n, s}(x, y)$ with $\partitions'_{n,s}(x, y)$.

Next, we show that we can also replace $\sup$ with $\max$. Define
\begin{align*}
B &:= \{(x_1,\dots,x_{n-1})\mid 
x\leq x_1 \leq \cdots \leq x_{n-1}\leq y-s, \text{ and } x_i+s\leq x_{i+1}\enspace \forall i\in [n-2]\}. 
\end{align*}
Intuitively, for each $i\in [n-1]$ the point $x_i$ is the right endpoint of the $i$-th interval in a partition where every two intervals are separated by exactly $s$ (so that the $(i+1)$-st interval starts at $x_i+s$).
Let $g_i: B\to[0,1]$ be a function defined by
\begin{align*}
g_i(x_1,&\dots,x_{n-1}) := \min\{v_i(x,x_1),v_i(x_1+s,x_2),\dots,v_i(x_{n-1}+s,y)\}.
\end{align*}
Since $B$ is a closed and bounded subset of $\mathbb{R}^{n-1}$, the Heine--Borel theorem implies that it is compact.
Moreover, since 
$v_i$ is the integral of an integrable function,
each of the $n$ functions inside the minimum operator is continuous.
This means that $g_i$, which is a minimum of continuous functions, is continuous too.
Hence, the extreme value theorem for functions of several variables implies that $g_i$ attains a maximum in $B$.
It follows that
\begin{align*}
\mms^{n,s}_i(x, y) 
&= \sup_{\partition\in \partitions'_{n,s}(x, y)}\min_{j\in [n]}v_i(P_j) \\
&= \sup_{(x_1,\dots,x_{n-1})\in B}g_i(x_1,\dots,x_{n-1}) \\
&=
\max_{(x_1,\dots,x_{n-1})\in B}g_i(x_1,\dots,x_{n-1}) \\
&=
\max_{\partition\in \partitions'_{n,s}(x, y)}\min_{j\in [n]}v_i(P_j),
\end{align*}
as claimed.
\end{proof}

From now on, we will work with this new definition of the maximin share.
We say that an $s$-separated partition is a \emph{maximin partition} for agent $i$ if every piece in the partition yields value at least $\mms^{n,s}_i$. 
Proposition~\ref{prop:alt-def} implies that every agent has at least one maximin partition.
Similarly, an $s$-separated allocation $\allocation = (A_1, \dots, A_n)$ is said to be an \emph{MMS-fair allocation} if $v_i(A_i)\ge \mms^{n,s}_i$ for each $i\in N$.

\section{Cake Cutting}
\label{sec:cake}

In this section, we consider cake cutting with separation, both in the Robertson--Webb query model and in a model where the agents' valuations are given explicitly.

\subsection{Robertson--Webb Query Model}

We begin by showing that the maximin share is an appropriate fairness criterion in our setting: it is always possible to fulfill this criterion for every agent using a quadratic number of queries in the Robertson--Webb model.
Our algorithm is similar to the famous Dubins--Spanier protocol for finding proportional allocations when separation is not required \citep{DubinsSp61}: we process the cake from left to right and, at each stage, allocate a piece of cake to an agent who demands the smallest piece.

\begin{theorem}
\label{thm:mms-algo}
For any instance of cake cutting with separation, there exists an MMS-fair allocation.
Moreover, given the maximin share of each agent, such an allocation can be computed using $O(n^2)$ queries in the Robertson--Webb model.
\end{theorem}

\begin{proof}
If there are at least two agents present, we ask each agent $i$ to mark the leftmost point $x_i$ such that $v_i(0,x_i) = \mms_i$. 
The agent who marks the leftmost $x_i$ is allocated the piece $[0,x_i]$ (with ties broken arbitrarily); we then remove this agent along with the piece $[x_i,x_i+s]$, and recurse on the remaining agents and cake.
If there is only one agent left, that agent receives all of the remaining cake.
Since we make $n-j$ \textsc{Cut} queries when there are $n-j$ agents left (and no \textsc{Eval} queries), our algorithm uses $\sum_{j=0}^{n-2}(n-j) = O(n^2)$ queries.

We now prove the correctness of the algorithm.
Consider any agent $i$ and her maximin partition $\partition=\{P_1, \dots, P_n\}$. 
If agent $i$ receives the first piece allocated by the algorithm, she receives value $\mms_i$. 
Else, the allocated piece is no larger than $P_1$.
Since the algorithm inserts a separator of length exactly $s$, the right endpoint of the first separator is either the same or to the left of the leftmost point of $P_2$.
Applying a similar argument repeatedly, we find that if agent $i$ is not allocated any of the first $n-j$ pieces, where $j\in [n-1]$, then the remaining cake contains the last $j$ pieces of $\partition$. In particular, for $j=1$ it follows that if agent $i$ receives the very last piece, 
she receives a value of at least $\mms_i$ in this case, too.
\end{proof}

The algorithm in Theorem~\ref{thm:mms-algo} crucially relies on knowing the maximin share of each agent.
Unfortunately, we show next that this knowledge is impossible to achieve in finite time, even if the valuations are piecewise constant but are not given explicitly as part of the input.
Our result is similar in spirit to the non-finiteness results for connected envy-free cake cutting \citep{Stromquist08}, equitability with connected pieces \citep{CechlarovaPi12,branzei2017query} and with arbitrary pieces \citep{procaccia2017lower}, and
average-proportionality \citep{segal2019families}.
However, all previous impossibility results were for two or more agents with possibly different valuations. 
In contrast, our impossibility result is attained even for a \emph{single} agent who wants to cut the cake into two $s$-separated pieces.

\begin{theorem}
\label{thm:mms-impossibility}
For each $s>0$, there is no finite algorithm (i.e., an algorithm that always terminates) that, given $n\in\mathbb N$ and $i\in N$, is guaranteed to compute $\emph{MMS}_i^{n,s}$
by asking agent $i$ a number of \textsc{Cut} and \textsc{Eval}
queries.
This holds even when $n=2$ and  agent $i$'s valuation is piecewise constant and strictly positive (but is not given explicitly).
\end{theorem}
We prove the theorem by reducing from a more general problem, which may be of independent interest.

\begin{framed}
\noindent
{Problem \findsum{s}, where $s\in[0,1)$ is a real parameter.}

\paragraph{Input:} $g: [0,1]\to [0,1]$, a continuous monotonically increasing bijective function, specified by oracles that can answer two kinds of queries:
\begin{itemize}
\item Given $x\in[0,1]$, what is $g(x)$?
\item Given $\alpha\in[0,1]$, what is $g^{-1}(\alpha)$?
\end{itemize}

\paragraph{Output:} A point $x_0\in[0,1-s]$ for which $g(x_0) +g(x_0+s)=1$.
\end{framed}
The function $h(x)=g(x)+g(x+s)$ is continuous, 
monotonically increasing, and satisfies $h(0)=g(0)+g(s)\le 1$, 
$h(1-s)=g(1-s)+g(1)\ge 1$. Therefore, 
\findsum{s} always has a solution due to the intermediate value theorem.
Further, \findsum{0} can be solved by a single query, namely, $g^{-1}(1/2)$.
Our lemma below shows that, for any $s>0$, \findsum{s} cannot be solved by finitely many queries. To prove Theorem~\ref{thm:mms-impossibility},
we then show that \findsum{s} can be reduced to computing $\mms^{2,s}_i$. 

\begin{lemma}
\label{lem:function-impossibility}
For each $s>0$, there is no algorithm that solves \findsum{s} using finitely many queries. This holds even if it is known that $g$ is piecewise linear.
\end{lemma}
\begin{proof}
Assume for contradiction that a finite algorithm exists.
We will show how an adversary can answer the queries made by the algorithm in such a way that after any finite number of queries, for any value of $x_0$ that the algorithm may output, there exists a piecewise linear function $g$ consistent with the adversary's answers such that 
$g(x_0)+g(x_0+s)\neq 1$.
This is sufficient to obtain the desired contradiction.

During the run of the algorithm, there is always a finite set of points $x\in[0,1]$ for which the algorithm knows the value of $g(x)$; we say that such points are \emph{recorded}. 
Initially, only points $0$ and $1$ are recorded:
since $g$ is a bijection and it is monotonically increasing, we have $g(0)=0$ and $g(1)=1$.
Given a point $x\in[0,1]$, we denote by $x_-$ the largest recorded point that is at most $x$, and by $x_+$ the smallest recorded point that is at least $x$.
If $x$ itself is recorded, then $x_-=x_+=x$.

When asked ``Given $x$, what is $g(x)$?'', if $x$ is recorded then the adversary replies $g(x)$; else, the adversary chooses a value for $g(x)$ so as to satisfy the following properties:
\begin{enumerate}[label=(\roman*)]
\item Monotonicity is preserved, i.e., $g(x_-) < g(x) < g(x_+)$.
\item If 
the point $x+s$ is recorded, then $g(x)\neq 1-g(x+s)$.
\item If 
the point $x-s$ is recorded, then $g(x)\neq 1-g(x-s)$.
\end{enumerate}
Since condition (i) allows infinitely many values to choose from, and each of the conditions (ii) and (iii) rules out at most one value, the adversary can make a choice satisfying these conditions.

When asked ``Given $\alpha$, what is $g^{-1}(\alpha)$?'', if there is a recorded point $x$ such that $g(x) = \alpha$,
then the adversary replies $x$;
else, the adversary chooses a point $x$ so as to satisfy the following properties:
\begin{enumerate}[label=(\roman*)]
\item Monotonicity is preserved, i.e.,
$g(x_-)<\alpha<g(x_+)$.
\item None of the points $x-s$, $x+s$, and $x$ is recorded. 
\end{enumerate}
Again, condition~(i) allows infinitely many points, and condition~(ii) rules out only a finite number of them.

Since the algorithm is finite, it must eventually return some number $x_0 \in[0,1-s]$.
Now, in order to falsify the algorithm's answer, the adversary voluntarily answers two more queries by the same rules as above: $g(x_0)$ and $g(x_0+s)$.
Now both $x_0$ and $x_0+s$ are recorded. The answering rules guarantee that $g(x_0)+g(x_0+s)\neq 1$, so $x_0$ cannot be a correct solution to \findsum{s}.

Finally, to complete the function $g$, the adversary simply connects every pair of consecutive recorded points linearly.
\end{proof}

\newcommand{\mmsalg}{\textsc{Alg}}
\begin{proof}[Proof of Theorem \ref{thm:mms-impossibility}]
Let $s>0$ be the separation parameter. 
We reduce \findsum{s} to the problem of computing $\mms_i^{2,s}$ for $i\in \{1, 2\}$:
we show that, given an algorithm \mmsalg{} that computes $\mms_i^{2,s}$ using a finite number of \textsc{Cut} and \textsc{Eval} queries in the Robertson--Webb framework, we can solve \findsum{s} for any function $g$ using finitely many queries.
This is sufficient by Lemma~\ref{lem:function-impossibility}.
To implement the reduction, we need to explain (1) how to answer
the \textsc{Eval} and \textsc{Cut} queries using the oracles for 
$g$ and $g^{-1}$, and (2) how to transform the output of \mmsalg{} into an answer to \findsum{s}. 

We approach (1) by imitating a valuation $v$ for which $v(x, y)=g(y)-g(x)$.
Specifically:
\begin{itemize}
\item If \mmsalg{} asks $\eval_i(x,y)$, we ask $g(x)$ and $g(y)$ and return $g(y)-g(x)$.
\item If \mmsalg{} asks $\cut_i(x,\alpha)$, we ask $g(x)$ and $g^{-1}(\alpha+g(x))$ and return the latter value.
\end{itemize}
At some point, \mmsalg{} stops and outputs some $r\in[0,1]$ as the value of $\mms_i^{2,s}$. At that point we ask one more query $g^{-1}(r)$ and return its result $x_0$ as the answer to \findsum{s}.

As \mmsalg{} is presumed to be correct, it must be the case that $\mms_i^{2,s}=r$ for any valuation $v$ that is compatible with the query replies. This means that there exists an $s$-separated partition of the cake into two pieces $[0,x_0]$ and  $[x_0+s,1]$ for which $v(0,x_0)=v(x_0+s,1)=r$.
Consequently, we have $v(0, x_0+s)=1-r$ and hence
$v(0,x_0) + v(0,x_0+s)=1$.
This is true, in particular, for the valuation $v(x, y)=g(y)-g(x)$, which was used to answer the queries. Hence,
$g(x_0)+g(x_0+s)=1$, so $x_0$ is indeed the right answer to \findsum{s}.
If \mmsalg{} used $t$ queries, then this answer was found using $2t+1$ queries to $g$ and~$g^{-1}$.
\end{proof}

Given a maximin partition of an agent, we can compute the agent's maximin share by simply taking the minimum among the agent's values for the pieces in the partition.
Moreover, for two agents with identical valuations, an allocation in which each agent receives at least their (common) maximin share corresponds to a maximin partition for the common valuation.
 Theorem~\ref{thm:mms-impossibility} therefore yields the following corollary,
 which also implies that an allocation whose existence is guaranteed by Theorem~\ref{thm:mms-algo} cannot be computed without the knowledge of the agents' maximin shares.

\begin{corollary}
\label{cor:mms-no-partition}
There is no finite algorithm in the Robertson--Webb model that can always 
(a) compute a maximin partition of a single agent, or 
(b) compute an MMS-fair allocation for $n$ agents.
This holds even when $n=2$ and the agents' valuations are piecewise constant (but not given explicitly).
\end{corollary}

Despite these negative results, we show next that it is possible to get an arbitrarily good approximation of the maximin share, partition, and allocation.

\begin{theorem}
\label{thm:mms-query-atleast}
Given an agent $i$ and a number $r > 0$, it is possible to decide whether $\emph{MMS}_i\geq r$ (and, if so, compute a partition $\partition$ such that $i$'s value for each part of $\partition$ is at least $r$) using at most $n$ queries in the Robertson--Webb model.
\end{theorem}

\begin{proof}
The idea is similar to that in the proof of Theorem~\ref{thm:mms-algo}.
Ask the agent to mark the leftmost point $x$ such that $v_i(0,x) = r$, make $[0,x]$ one of the pieces in a potential partition, add a separator $[x,x+s]$, and repeat starting from $x+s$.
If there is still value at least $r$ left after $n-1$ iterations, answer Yes; else, answer No. 
It is clear that at most $n$ queries are used.

If the algorithm answers Yes, then it finds a partition with value at least $r$, so $\mms_i\geq r$.
Conversely, suppose that $\mms_i\geq r$, and consider a maximin partition $\partition = \{P_1, \dots, P_n\}$. The right endpoint of $P_1$ is either the same or to the right of our first marked point $x$.
In addition, since our algorithm inserts a separator of length exactly $s$, the right endpoint of our first separator is no further to the right than the left endpoint of $P_2$.
Applying a similar argument $n-1$ times, we find that the right endpoint of our $(n-1)$-st separator is no further to the right than the left endpoint of $P_n$.
Since the value of $P_n$ is at least $\mms_i$, our remaining piece also has value at least $\mms_i\geq r$. Hence the algorithm answers Yes, as claimed.
\end{proof}

As an example, suppose that $v_i(0,1/3) = 0.4$, $v_i(1/3,2/3) = 0$, $v_i(2/3,1) = 0.6$, and the value is distributed uniformly within $[0, 1/3]$ and within $[2/3, 1]$.
Moreover, $s=1/3$ and we need to decide whether $\mms^{2,s}_i \ge 0.4$.
The algorithm in Theorem~\ref{thm:mms-query-atleast} first marks $x = 1/3$, then adds a separator $[1/3, 2/3]$, and finally answers Yes because the value of the remaining cake is $0.6 \ge 0.4$.

Combining the algorithm in Theorem~\ref{thm:mms-query-atleast} with binary search over possible values of $r$ enables us to approximate the maximin share.

\begin{corollary}
\label{cor:mms-query-bound}
Given an agent $i$ and a number $\varepsilon > 0$, it is possible to find a number~$r$ for which $\emph{MMS}_i-\varepsilon\leq r\leq \emph{MMS}_i$ (together with a partition whose min-value for agent~$i$ is at least $r$) using $O(n\log(1/\varepsilon))$ Robertson--Webb queries.
\end{corollary}

If instead of the exact $\mms_i$, we are given a number $r_i\leq \mms_i$ for each agent~$i$, the algorithm for computing an MMS-fair allocation in Theorem~\ref{thm:mms-algo} still computes an allocation in which agent $i$ receives value at least $r_i$; the proof is essentially the same.
We therefore have the following corollary.

\begin{corollary}
\label{cor:mms-allocation-approx}
For any $\varepsilon > 0$, it is possible to compute an allocation in which every agent $i$ receives value at least $\emph{MMS}_i-\varepsilon$ using 
$O(n^2\log(1/\varepsilon))$ Robertson--Webb queries.
\end{corollary}

Next, we consider the question of deciding whether the maximin share of an agent is \emph{strictly greater than} a given number $r$.
At first glance, it may seem that, to answer this question, we can simply run the algorithm from Theorem~\ref{thm:mms-query-atleast} and answer Yes exactly when the value left after $n-1$ iterations is strictly greater than $r$.
While this modification indeed works if the density function is positive on the entire cake, it may fail when intervals with zero value are present.
Concretely, suppose again that $v_i(0,1/3) = 0.4$, $v_i(1/3,2/3) = 0$, $v_i(2/3,1) = 0.6$, and the value is distributed uniformly within $[0, 1/3]$ and within $[2/3, 1]$.
Moreover, $s=1/3$, but this time we need to decide whether $\mms^{2,s}_i > 0.4$.
When we run the modified algorithm, there is value $0.6 > 0.4$ left after $n-1 = 1$ iterations, which may suggest that the maximin share can exceed $0.4$.
However, this is not the case, because of 
the zero-valued middle part.

This example suggests a simple modification to the algorithm from Theorem~\ref{thm:mms-query-atleast}: instead of marking the leftmost point such that the value of the resulting piece is $r$, we should mark the \emph{rightmost} point with this property. It is not difficult to verify that $\mms_i>r$ if and only if after executing this modified algorithm we are left with a positive-value piece.
The modified algorithm requires a query $\cutr_i(x,\alpha)$ that returns the rightmost point $y$ for which that $v(x,y) = \alpha$.
Unfortunately---and perhaps surprisingly---$\cutr$ \emph{cannot} be implemented using the queries available via the Robertson--Webb model---see 
Appendix~\ref{sec:cutright}.

Nevertheless, we can get around this issue by flipping the cake, i.e., going over the cake from right to left. 
E.g., in the first iteration, we need to identify the leftmost point~$x$ such that $v_i(x,1)=r$. This is equivalent to finding the leftmost point~$x$ such that $v_i(0,x)=1-r$, which can be done with a standard $\cut_i(0,1-r)$ query.

\begin{theorem}
\label{thm:mms-query-greater}
Given an agent $i$ and a number $r \ge 0$, it is possible to decide whether $\emph{MMS}_i> r$ using at most $2n-1$ queries in the Robertson--Webb model.
\end{theorem}

\begin{proof}
Ask the agent to mark the leftmost point $x$ such that $v_i(0,x)=1-r$, make $[x,1]$ one of the pieces in a potential partition, add a separator $[x-s,x]$, ask the agent to mark the leftmost point $x'$ such that $v_i(0,x') = v_i(0, x-s) - r$, make $[x',x-s]$ another piece in the potential partition, add a separator $[x'-s,x']$, and repeat going from $x'-s$ leftwards.
If there is still cake left after $n$ pieces have been created, answer Yes; else, answer No.

If there is no cake left once we have created $n$ pieces (or if we cannot create $n$ pieces at all), then by an argument similar to those in Theorems~\ref{thm:mms-algo} and~\ref{thm:mms-query-atleast}, the maximin share cannot be greater than $r$, so the answer No is correct.

Suppose now that there is cake left after $n$ pieces have been created;
let the resulting partition be $\partition=\{P_1, \dots, P_n\}$, where as usual $P_1,\dots,P_n$ are ordered from left to right on the cake (in particular, $P_1$ was the last created piece).
We will show that $\partition$ can be modified so that the min-value of the resulting partition $\partition'=\{P_1', \dots, P_n'\}$ is strictly more than~$r$, thereby proving that the answer Yes is correct.

Suppose $P_1=[x_1, y_1]$; by our assumption, $x_1>0$ and $v_i(0, x_1)=\varepsilon_1>0$.  
We set $P'_1=[0, y'_1]$, where $v_i(y'_1, y_1)=\varepsilon_1/2$:
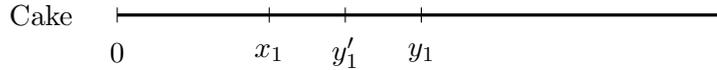
\vspace{2mm}
\begin{center}
\begin{tikzpicture}[scale=1]
\draw (2,4.9) -- (2,5.1);
\node at (2,4.5) {$0$};
\draw (4,4.9) -- (4,5.1);
\node at (4,4.5) {$x_1$};
\draw (5,4.9) -- (5,5.1);
\node at (5,4.5) {$y_1'$};
\draw (6,4.9) -- (6,5.1);
\node at (6,4.5) {$y_1$};
\draw[very thick] (2,5) -- (10,5);
\node at (1,5) {Cake};
\end{tikzpicture}
\end{center}
This ensures that 
\vspace{-2mm}
$$
v_i(P'_1)=v_i(P_1)+v_i(0, x_1)-v_i(y'_1, y_1) = v_i(P_1)+\varepsilon_1/2 > v_i(P_1).
$$
On the other hand, since $v_i(y'_1, y_1)>0$, we have $y'_1<y_1$. We can now shift the second interval of the partition in a similar manner: if $P_2=[x_2, y_2]$,
we set $P'_2=[x'_2, y'_2]$, where
$x'_2=y'_1+s < y_1+s\le x_2$, $\varepsilon_2 = v_i(x'_2, x_2)>0$ (where the inequality holds by our choice of $x_2$), 
and $y'_2$ satisfies $v_i(y'_2, y_2)=\varepsilon_2/2$. Again, this ensures that 
$v_i(P'_2)=v_i(P_2)+\varepsilon_2/2 >v_i(P_2)$ and $y'_2 < y_2$.
We repeat this process for all subsequent pieces, thereby ensuring that
$\{P'_1, \dots, P'_n\}$ is $s$-separated and $v_i(P'_j)>v_i(P_j)\ge r$ for all $j\in [n]$,
as claimed.
\end{proof}

Theorems~\ref{thm:mms-query-atleast} and \ref{thm:mms-query-greater} immediately imply the following corollary.

\begin{corollary}
\label{cor:mms-query-equal}
Given an agent $i$ and a number $r \ge 0$, it is possible to decide whether $\emph{MMS}_i= r$ using at most $3n-1$ queries in the Robertson--Webb model.
\end{corollary}

\subsection{Explicit Piecewise Constant Valuations}

For our next result, we assume that 
all agents have piecewise constant valuations and, moreover, these valuations are given \emph{explicitly}. That is, for an agent $i\in N$ we are given a list of breakpoints $(p_0, p_1, \dots, p_d)$ with $p_0=0$, $p_d=1$, and a list of densities $(\gamma_1, \dots, \gamma_d)$, so that for each $j\in [d]$ and each $x\in [p_{j-1}, p_j]$ the valuation function $v_i$ of agent $i$ satisfies $v_i(0,x) = \sum_{\ell=1}^{j-1}[\gamma_\ell\cdot (p_\ell-p_{\ell-1})]+\gamma_j\cdot (x-p_{j-1})$.
Moreover, all breakpoints and densities are rational numbers, represented as 
fractions whose numerators and denominators are given in binary. It is straightforward to implement both types of Robertson--Webb queries in this model, so every problem that can be solved in polynomial time in the Robertson--Webb model can also be solved in polynomial time in this model.
But the explicit representation offers additional benefits: we can compute the agents' maximin shares \emph{exactly} rather than approximately.

\begin{theorem}\label{thm:mms-exact}
Given an agent $i$ with a piecewise constant valuation function given explicitly, we can compute $\emph{MMS}^{n,s}_i$ in time polynomial in the size of the input. 
\end{theorem}

At a high level, the proof of Theorem~\ref{thm:mms-exact} proceeds by formulating a linear program whose solution corresponds to $\mms_i$. 
The challenge is that in order to have a linear program that returns a correct answer, we need to find out the intervals to which each endpoint of a maximin partition belongs.
To accomplish this, we proceed from left to right, determining the interval for one endpoint at a time.
By comparing the maximin shares between optimal subpartitions to the left and right of the potential intervals to which the next endpoint belongs, we ensure that at each step of the algorithm, there exists a maximin partition whose endpoints are consistent with the intervals we have chosen.
The full details of the algorithm are rather involved; we defer them to Appendix~\ref{app:LP-proof}.

Combined with Theorem~\ref{thm:mms-algo}, Theorem~\ref{thm:mms-exact} implies that when agents have piecewise constant valuations given explicitly, an MMS-fair allocation can be computed efficiently (cf. Corollary~\ref{cor:mms-no-partition}).

\begin{corollary}
\label{cor:mms-allocation-exact}
For agents with piecewise constant valuations given explicitly, an MMS-fair allocation can be computed in time polynomial in the size of the input.
\end{corollary}

\section{Pie Cutting}
\label{sec:pie}
In the canonical model of cake cutting, the cake is assumed to be linear. By contrast, in this section we assume that it is circular.
In other words, our resource is represented by the interval $[0,1]$ with its two endpoints identified with each other.
The respective division problem is known in the literature as \emph{pie cutting}; its applications include dividing the shoreline of an island among its inhabitants and splitting a daily cycle for using a facility \citep{thomson2007children,BramsJoKl08,BarbanelBrSt09}.

The definitions of $s$-separated partitions and allocations can be readily adjusted to pie cutting---the only difference is that, due to the circular structure, there are $n$ separators in pie cutting rather than $n-1$ (so we assume that $s < 1/n$).
Note that, since the pie is one-dimensional, distances are measured along the circumference of the pie.
We denote by $\partitions_{n,s}$ the set of $s$-separated partitions with respect to the pie, and by $\partitions'_{n, s}\subseteq \partitions_{n, s}$ the subset of partitions for which every pair of consecutive pieces is separated by length exactly $s$.
We number the parts in a partition in the clockwise order starting from $0$: that is, point $0$ is either contained in the first part or in the separator between the $n$-th part and the first part.
The maximin share can then be defined similarly to how it is defined in cake cutting (Definition~\ref{def:MMS}); just as in Proposition~\ref{prop:alt-def}, we can show that $\mms^{n,s}_i = \max_{\partition\in\partitions'_{n, s}}\min_{j\in [n]}v_i(P_j)$.

However, in pie cutting, unlike in cake cutting, an MMS-fair allocation does not necessarily exist.
This is evident in the example in Figure~\ref{fig:pie-no-mms}, where $1/4 < s < 1/2$, $0 < \varepsilon < \min\{s-1/4,1/2-s\}$, and Alice values the pieces of length $\varepsilon$ centered at the top and bottom of the pie at $1/2$ each, while Bob values similar pieces on the left and right at $1/2$ each.
Since $s < 1/2-\varepsilon$, the maximin share of each agent is $1/2$.
However, since the distance between any point in Alice's piece and any point in Bob's piece is at most $1/4 + \varepsilon < s$, no $s$-separated allocation gives both agents a positive value.
Hence, no MMS-fair allocation exists. 

\begin{figure}[!ht]
\centering
\begin{tikzpicture}[scale=0.8]
\draw (3,3) circle [radius = 2];
\draw (8.5,3) circle [radius = 2];
\draw [ultra thick,domain=80:100] plot ({3+2*cos(\x)}, {3+2*sin(\x)});
\draw [ultra thick,domain=260:280] plot ({3+2*cos(\x)}, {3+2*sin(\x)});
\draw [ultra thick,domain=-10:10] plot ({8.5+2*cos(\x)}, {3+2*sin(\x)});
\draw [ultra thick,domain=170:190] plot ({8.5+2*cos(\x)}, {3+2*sin(\x)});
\node at (3,3) {Alice};
\node at (8.5,3) {Bob};
\node at (3,4.7) {$\varepsilon$};
\node at (3,1.3) {$\varepsilon$};
\node at (6.8,3) {$\varepsilon$};
\node at (10.2,3) {$\varepsilon$};
\end{tikzpicture}
\caption{Example of a pie cutting instance with no MMS-fair allocation. Each of the two agents uniformly values the bold part of the pie, $s>1/4$, and $0<\varepsilon < \min\{s - 1/4, 1/2-s\}$.}
\label{fig:pie-no-mms}
\end{figure}
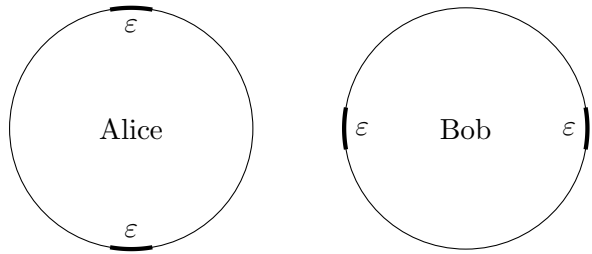

In fair division, a common response to the non-existence of MMS-fair allocations is to seek allocations that guarantee each agent a constant fraction of her maximin share.
However, the same example shows that, in our setting, no positive fraction of the maximin share can be guaranteed.
Given this negative result, one may wonder whether any meaningful fairness guarantee can be achieved in pie cutting with separation. Fortunately, positive results can be obtained if, instead of using \emph{cardinal} approximations of the maximin share, we relax this criterion in an \emph{ordinal} manner.
Specifically, when each agent computes her maximin share, we ask her to pretend that there are $k>n$ agents and divide the pie into $k$ pieces; intuitively, this results in smaller pieces and hence a less ambitious benchmark.
We refer to the resulting notion as the \emph{$1$-out-of-$k$} maximin share and write $\mms^{k,s}_i$ or simply $\mms^k_i$ for the share of agent $i$.
Ordinal approximations were introduced by \citet{Budish11} in the context of indivisible item allocation.\footnote{We note that $1$-out-of-$k$ maximin share is a special case of the $\ell$-out-of-$k$ maximin share notion introduced by \citet{BabaioffNiTa21} and further studied by \citet{Segalhalevi20}, which selects $\ell$ pieces of minimum value from a partition into $k$ pieces (see also \Cref{app:l-out-of-k}).} 
In particular, he considered the case $k=n+1$. 

It turns out that this relaxation is precisely what we need for pie cutting.

\begin{theorem}
\label{thm:pie-mms-algo}
For any pie cutting instance with $n$ agents, there exists an allocation in which every agent $i$ receives a piece of value at least $\emph{MMS}^{n+1}_i$.
Moreover, given the $1$-out-of-$(n+1)$ maximin share of each agent, such an allocation can be computed using $O(n^2)$ queries in the Robertson--Webb model.
\end{theorem}

The idea behind our algorithm is similar to that of the analogous result for cake cutting (Theorem~\ref{thm:mms-algo}). Note, however, that in case of a pie there is no natural starting point; in particular, by starting at $0$, we may destroy one of the pieces in each agent's partition. This is why we need $n+1$ pieces in the partition rather than $n$.

\begin{proof}
We ask each agent $i$ to mark the leftmost point $x_i$ (i.e., the first such point when moving clockwise from $0$) such that $v_i(0,x_i) = \mms^{n+1}_i$. 
The agent who marks the leftmost $x_i$ is allocated the piece $[0,x_i]$ (with ties broken arbitrarily); we then remove this agent along with the piece $[x_i,x_i+s]$, and recurse on the remaining agents and pie.
If there is only one agent left, we still allocate to that agent a piece worth $\mms^{n+1}_i$ (rather than the entire remaining pie).
Since we make $n-j$ \textsc{Cut} queries when there are $n-j$ agents left (and no \textsc{Eval} queries), our algorithm uses $\sum_{j=0}^{n-1}(n-j) = O(n^2)$ queries.

We now prove the correctness of the algorithm.
Consider any agent $i$ and her $1$-out-of-$(n+1)$ maximin partition $\partition=\{P_1, \dots, P_{n+1}\}$.
For $j\in [n]$, let $Q_j=P_{j+1}$ if $0$ is in the interior of $P_1$, and $Q_j=P_j$ otherwise; we
will write $Q_j=[x_j, y_j]$.
Note that the segment $[0, y_j]$ contains $j$ parts of $\partition$.
If agent $i$ receives the first piece allocated by the algorithm, she receives value $\mms^{n+1}_i$.
Else, the right endpoint of the allocated piece is no further to the right than $y_1$.
Since the algorithm inserts a separator of length exactly~$s$, the right endpoint of the first separator is no further to the right than $x_2$.
Applying a similar argument repeatedly, we find that if agent $i$ is not allocated any of the first $n-1$ pieces, then after removing the $(n-1)$-st piece and the following separator, the remaining cake contains $Q_n$. 
Now, if $0$ is in the interior of $P_1$, then the remaining cake also contains a positive amount of $P_1$ as well as the separator between $P_{n+1}=Q_n$ and $P_1$. 
On the other hand, if $0$ is not in the interior of $P_1$, then $Q_n=P_n$ and the remaining cake contains $P_{n+1}$ as well as the separator between $P_n$ and $P_{n+1}$. In either case, the remaining cake contains $Q_n$ as well as the separator that comes after $Q_n$.
Hence, if we allocate a piece of value $\mms^{n+1}_i$ to $i$, its right endpoint is no further to the right than $y_n$ and thus the remaining cake contains an unallocated segment of length $s$, which will serve as a separator between the piece that was allocated first and the piece that was allocated last. 
It follows that in either case the resulting allocation is $s$-separated.
\end{proof}

In \Cref{app:l-out-of-k}, we present a generalization of Theorem~\ref{thm:pie-mms-algo} using the \emph{$\ell$-out-of-$k$ maximin share} notion of \citet{BabaioffNiTa21}.

Recall that for cake cutting there exists an algorithm that, given an agent $i$ and a number $r$, decides whether $\mms_i > r$ and 
whether $\mms_i = r$ (Theorem~\ref{thm:mms-query-greater} and Corollary~\ref{cor:mms-query-equal}). In contrast, for pie cutting this is not the case.

\begin{theorem}
\label{thm:pie-mms-greater}
Fix any $k\geq 2$. For pie cutting, there is no finite algorithm in the Robertson--Webb model that can decide, for any agent $i$ and real number $r$, whether $\emph{MMS}^k_i > r$ or whether $\emph{MMS}^k_i = r$, even when the valuation of this agent is piecewise constant (but not given explicitly).
\end{theorem}

\begin{proof}
Assume for contradiction that such an algorithm exists, and take $r=\frac{1}{k}-s$.
We show how an adversary can answer the queries of the algorithm in such a way that after a finite number of queries, there exists a piecewise constant valuation function consistent with the answers for which $\mms^k_i > r$, but also one for which $\mms^k_i = r$.

The adversary records any point that appears in a query or in its own answer, and answers queries 
as if the valuation is uniform throughout the pie---that is, for any two consecutive recorded points on the pie, if the interval between them has length $t$, then it also has value $t$.
Suppose that some finite number of queries have been answered in this manner, and consider the following two possibilities:

\underline{Possibility 1}: The entire valuation function is uniform. 
For any $s$-separated partition~$\partition$, the sum of the lengths of the $k$ pieces in $\partition$ is at most $1-ks$, so one of the pieces has length (and value) at most $\frac{1-ks}{k}=r$.
On the other hand, there exists an $s$-separated partition $\partition'$ such that each of the $k$ pieces in $\partition'$ has length $r$.
Hence $\mms^k_i = r$ in this case.

\underline{Possibility 2}: Consider all $s$-separated $k$-partitions in which each of the $k$ pieces has length $r$.
Among all such partitions, choose a partition $\partition$ for which none of the $2k$ endpoints of the pieces coincides with any recorded point; since there are infinitely many partitions and only a finite number of them are forbidden, this choice is possible.
If a piece or a separator does not contain a recorded point, record an arbitrary point in its interior.
This ensures that each interval between two recorded points contains at most one endpoint of $\partition$. Now, for each interval $I$ of length $z$ containing an endpoint of $\partition$,
distribute a value of $z$ uniformly within the intersection of $I$ with the associated piece of $\partition$, so the intersection of $I$ and the associated separator has value zero.
For the remaining intervals, their value (which is equal to their length) is distributed uniformly within the interval.
The resulting valuation function is piecewise constant, and each of the $k$ pieces of $\partition$ has value strictly greater than $r$.
Hence $\mms^k_i > r$.

We conclude that a finite algorithm cannot distinguish between the case $\mms^k_i>r$ and the case $\mms^k_i=r$.
\end{proof}

Theorem~\ref{thm:pie-mms-greater} leaves open the question of whether it is possible to decide whether $\mms^k_i\geq r$ for a given $r$.
We show next that the answer to this question, too, is negative.
We do so by reducing from the following problem, which may be of independent interest. 

\newcommand{\haslowval}[1]{\textsc{HasLowValue}($#1$)}
\begin{framed}
\noindent
{Problem \haslowval{s,q}, where $s, q\in[0,1)$.}

\paragraph{Input:} A valuation function $v$ on a pie $[0,1]$, accessible through $\cut$ and $\eval$ queries.

\paragraph{Output:} Yes if the pie contains an interval of length $s$ with value at most $q$, i.e., there is an $x_0\in[0,1]$ for which $v(x_0,x_0+s)\leq q$, where $x_0+s$ is computed modulo~$1$.
\end{framed}
\begin{lemma}
\label{lem:pie-separator}
For any real numbers $s, q$ with $s>q>0$, 
there is no algorithm 
that solves $\haslowval{s,q}$
using finitely many queries
in the Robertson--Webb model.
This holds even if the valuation $v$ is known to be piecewise constant and strictly positive (but not given explicitly).
\end{lemma}
\begin{proof}
Suppose for contradiction that there exists an algorithm as in the theorem statement. 
Assume without loss of generality that the algorithm only asks queries of the form $\eval(0,x)$ and $\cut(0,\alpha)$---the queries $\eval(x,y)$ and $\cut(x,\alpha)$ can be easily simulated using the former two query types.

During the run of the algorithm, there is always a finite set of points $x\in [0,1]$ for which the algorithm knows the value of $v(0,x)$; we say that such points are \emph{recorded}.
Initially only point $0$ (which is equivalent to point $1$) is recorded.
Given a point $x$, we denote by $x_-$ the closest recorded point counterclockwise, and by $x_+$ the closest recorded point clockwise. If $x$ itself is recorded, then $x_- = x_+ = x$.

We will show how an adversary can answer the queries made by the algorithm.
For any integer $t\ge 0$, after $t$ queries are made, the adversary has in mind a valuation function~$v_t$ satisfying the following properties:
\begin{enumerate}[label=(\roman*)]
\item $v_t$ is compatible with all previously recorded points.
\item There is a point $x_t$ such that neither $x_t$ nor $x_t+s$ is recorded, the density between $x_t$ and $x_t+s$ is exactly $q/s$ throughout, and the density elsewhere is strictly greater than $q/s$ throughout.
\end{enumerate}
For the base case $t=0$, the adversary chooses any $x_0\not\in\{0,1-s\}$, and sets the density outside $[x_0,x_0+s]$ to $(1-q)/(1-s)$. 
Observe that the resulting valuation is normalized,  
and $(1-q)/(1-s)$ is greater than $q/s$ since $q < s$.

For $t\ge 1$, assume that the adversary has in mind the valuation $v_{t-1}$ and the associated point $y := x_{t-1}$.
Recall that neither $y$ nor $y+s$ is recorded.
If there is no recorded point in the range $(y, y+s)$, the adversary voluntarily records an arbitrary point in that range, and similarly for the complement range $(y+s, y)$.
This ensures that $y_+$ belongs to the interval $(y, (y+s)_-]$.
Note that none of the points $y_-$, $y_+$, $(y + s)_-$, and $(y + s)_+$ belongs to the set $\{y, y+s\}$.
When the algorithm makes the $t$-th query, if both $y$ and $y+s$ would still be unrecorded upon answering according to $v_{t-1}$, the adversary answers the query according to $v_{t-1}$ and sets $v_t = v_{t-1}$ and $x_t = y$.
Else, assume that $y$ would be recorded if the adversary answers the query according to $v_{t-1}$; the case where $y+s$ would be recorded can be handled analogously. The adversary will then try to ``shift'' the low-value region clockwise so that both of its endpoints remain unrecorded.

Specifically, let $\varepsilon = \min\{y_+-y, (y+s)_+-(y+s)\}$, where subtraction is modulo 1, 
and let $z=y+\varepsilon/2$. We then have $z\in (y, y_+)$, $z+s\in (y+s, (y+s)_+)$; this ensures that there are no recorded points between $y$ and $z$
and between $y+s$ and $z+s$.
The adversary will construct $v_t$ so that it has density $q/s$ in $[z, z+s]$; this requires increasing the density in $[y, z]$ and lowering the density in $[y+s, z+s]$. However, to remain consistent with the previous answers, 
the adversary should not change the value of the segments $[y_-, y_+]$ and $[(y+s)_-, (y+s)_+]$.
To accomplish these goals,
the adversary constructs $v_t$ by making 
the following changes to $v_{t-1}$:
\begin{enumerate}[label=(\alph*)]
\item 
Set the density throughout $[y_-,z]$ to be $v_{t-1}(y_-, z)/(z-y_-)$.
Then $v_t(y_-, z)=v_{t-1}(y_-, z)$ and hence
the value of $[y_-, y_+]$ is preserved:
$v_t(y_-, y_+)=v_{t-1}(y_-, y_+)$.
Note that this operation decreases the density in 
$[y_-,y]$ and increases the density in $[y,z]$.
\item Set the density in $[y+s,z+s]$ to be $q/s$, 
and the density in $[z+s,(y+s)_+]$ to be a constant $d$ such that
\[
v_{t-1}(y+s, (y+s)_+) = (z-y)\cdot \frac{q}{s} + ((y+s)_+ - (z+s))\cdot d.
\]
Then the total value of $[y+s, (y+s)_+]$ and hence of $[(y+s)_-, (y+s)_+]$ is preserved, i.e., 
$v_t((y+s)_-, (y+s)_+) = v_{t-1}((y+s)_-, (y+s)_+)$.
\end{enumerate}

\begin{figure}[!ht]
\centering
\begin{tikzpicture}[scale=1.17]
\draw (3,7) -- (3,3) -- (14,3);
\node at (4.6,2.5) {\footnotesize $y_-$};
\node at (5,2.5) {\footnotesize $y$};
\node at (6,2.5) {\footnotesize $z$};
\node at (7.8,2.5) {\footnotesize $y_+$};
\node at (9.5,2.5) {\footnotesize $(y+s)_-$};
\node at (11,2.5) {\footnotesize $y+s$};
\node at (12,2.5) {\footnotesize $z+s$};
\node at (13,2.5) {\footnotesize $(y+s)_+$};
\draw [ultra thick] (4.6,2.9) -- (4.6,3.1);
\draw (5,2.9) -- (5,3.1);
\draw (6,2.9) -- (6,3.1);
\draw [ultra thick] (7.8,2.9) -- (7.8,3.1);
\draw [ultra thick] (9.5,2.9) -- (9.5,3.1);
\draw (11,2.9) -- (11,3.1);
\draw (12,2.9) -- (12,3.1);
\draw [ultra thick] (13,2.9) -- (13,3.1);
\node at (14.3,3) {\footnotesize pie};
\node at (3,7.3) {\footnotesize density};
\draw (2.9,4) -- (3.1,4);
\node at (2.5,4) {\footnotesize $q/s$};
\draw[blue,thick] (5,4) -- (11,4);
\draw[blue,thick] (3,5) -- (5,5);
\draw[blue,thick] (11,5) -- (14,5);
\draw[red,dashed,ultra thick] (3,5) -- (4.6,5);
\draw[red,dashed,ultra thick] (13,5) -- (14,5);
\draw[red,dashed,ultra thick] (6,4) -- (12,4);
\draw[red,dashed,ultra thick] (4.6,4.29) -- (6,4.29);
\draw[red,dashed,ultra thick] (12,6) -- (13,6);
\end{tikzpicture}
\caption{Illustration of the construction in the proof of Theorem~\ref{lem:pie-separator}, focusing only on a portion of the pie.
The solid blue lines depict the density function of the valuation $v_{t-1}$ and the dashed red lines depict the density function of the valuation $v_t$.
Thicker tickmarks denote recorded points.
}
\label{fig:pie-impos-shifting}
\end{figure}
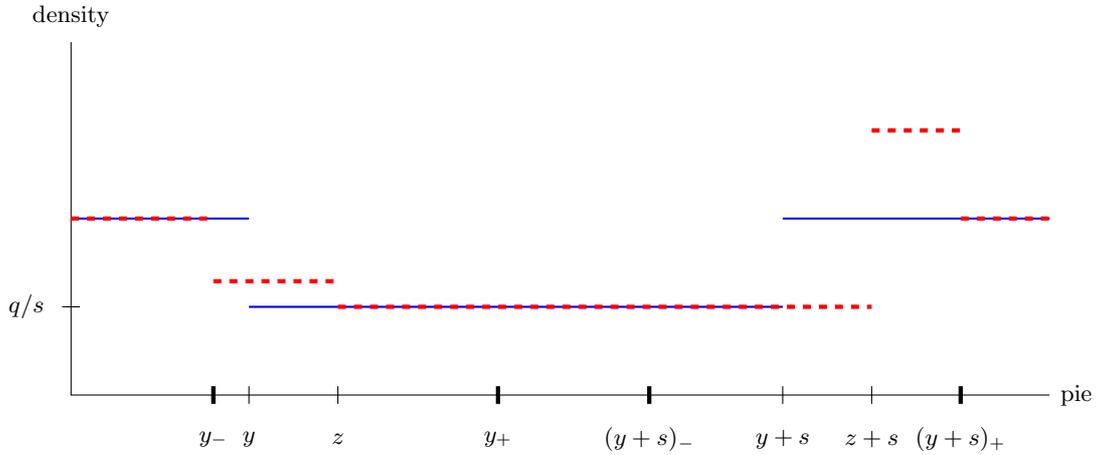

See Figure~\ref{fig:pie-impos-shifting} for an illustration of this construction.
Inductively, in $v_{t-1}$, the average density in the range $[y_-, y_+]$ is strictly greater than $q/s$.
Since the value of $[y_-, y_+]$ is preserved and the
density in $[z,y_+]$ is set to be exactly $q/s$ under $v_t$, 
the density in $[y_-,z]$ under $v_t$ is strictly greater than $q/s$.
Similarly, the constant density in $[z+s,(y+s)_+]$ in $v_t$ is strictly greater than $q/s$.

The adversary then answers the query according to $v_t$.
If the query is an \textsc{Eval} query, the new recorded point is $y\not\in\{z,z+s\}$.
Else, the query is of the form $\cut(0, \alpha)$. For $y$ to become recorded as a result of this query with respect to $v_{t-1}$, it would have to be the case that $y=0$ or $v_{t-1}(0, y)=\alpha$. We can rule out the former case, as $0$ is recorded at the beginning, and we know that the interval $[y_-, y_+]$ contains no recorded points.
In the latter case, observe that 
$[(y+s)_-, (y+s)_+]$ contains no recorded 
points either. So in particular neither 
of the intervals $I' =[y_-, z]$ and $I'' = [y+s, (y+s)_+]$
contains $0$.
Moreover, we have $v_t(I')=v_{t-1}(I')$, 
$v_t(I'') = v_{t-1}(I'')$,
and $v_t$ is identical to $v_{t-1}$ outside of $I'$ and $I''$.
Hence, it holds that $v_t(0, z) = v_{t-1}(0, z) > v_{t-1}(0, y) = \alpha$.
On the other hand, since transforming $v_{t-1}$ into $v_t$ lowers the density in $[y_-, y]$, 
we have $v_t(0, y)< v_{t-1}(0, y)=\alpha$. Combining these observations, 
we obtain $v_t(0, y) < \alpha < v_t(0, z)$.
Hence, if the adversary answers 
$\cut(0, \alpha)$ according to $v_t$, the resulting point will be located in the interval $(y,z)$, so in particular it will be different from $z$ and $z+s$.
This means that $v_t$ satisfies the properties (i) and (ii) with $x_t = z$.

Since the algorithm is finite, it must eventually answer whether there is an interval of length $s$ with value at most $q$.
If it answers No, the adversary reveals that the valuation is $v_t$, so the value of $[x_t,x_t+s]$ is exactly $q$.
On the other hand, if the algorithm answers Yes, the adversary constructs $v_{t+1}$ from $v_t$ by changing the density in $[(x_t)_-,(x_t)_+]$ to a constant so as to preserve the total value in this interval. By property (ii) of $v_t$, this constant density is strictly greater than $q/s$.
With respect to $v_{t+1}$, the density of any interval is at least $q/s$, and any interval of length $s$ contains a positive-length portion with density greater than $q/s$. 
This means that the value of any such interval is greater than $q$, so by revealing the valuation to be $v_{t+1}$, the adversary can again falsify the algorithm's answer. This completes the proof.
\end{proof}

\begin{remark}
The requirement $s>q>0$ is essential for the impossibility.
Indeed, when $q\geq s$, the answer 
to \haslowval{s,q} is always Yes, whereas
when $q=0$, \haslowval{s,q} can be answered using finitely many queries; the algorithm is similar to the one in the proof of Theorem~\ref{thm:pie-1/n-positive}.
\end{remark}

\begin{theorem}
\label{thm:pie-mms-greater-equal}
In pie cutting, for any $k\geq 2$, there is no finite algorithm in the Robertson--Webb model that can decide, for any agent $i$ and real number $r$, whether $\emph{MMS}^k_i \geq r$, even when the valuation of agent $i$ is piecewise constant and strictly positive (but not given explicitly).
\end{theorem}

\begin{proof}
We prove a somewhat stronger claim:
for any $r\in \left(\frac{1}{k}-s,\frac{1}{k}\right)$, there is no finite algorithm that can decide whether $\mms^k_i \geq r$.
In particular, there is an entire interval of ``undecidable values'' rather than a single such value.
For the sake of convenience, we represent the pie by the interval $[0,k]$ instead of $[0,1]$.

The proof is by reduction from 
\haslowval{s,q}.
We show that, given an algorithm \mmsalg{} that
decides, for any $r\in \left(\frac{1}{k}-s,\frac{1}{k}\right)$, whether $\mms^k_i \geq r$ on the pie $\pi_k := [0,k]$,
we can decide, for any $q\in(0,s)$ and any valuation $v$ on the pie $\pi_1 := [0,1]$, whether there exists an interval of length $s$ and value at most $q$.
This is sufficient by Lemma~\ref{lem:pie-separator}.

We run \mmsalg{} with $r := \frac{1}{k} - q$; note that $r\in \left(\frac{1}{k}-s,\frac{1}{k}\right)$.
For each query asked by \mmsalg{}, 
we reply like an agent whose density function on $\pi_k$ is made of $k$ copies of the density function of $v$ on $\pi_1$.

If there exists an interval in $\pi_1$ of length $s$ and value at most $q$, then by using the $k$ corresponding pieces of $\pi_k$ as separators, the resulting partition has min-value at least $(1-kq)/k = r$, 
so \mmsalg{} answers Yes.

Conversely, suppose that no such interval exists in $\pi_1$, and consider any $s$-separated partition of $\pi_k$.
Each separator takes up value strictly greater than $q$, so the remaining value outside the $k$ separators is less than $1-kq = kr$.
Hence, at least one of the pieces in the partition has value less than $r$, implying that $\mms^k_i < r$; so \mmsalg{} answers No. 

In both cases, the answer of \mmsalg{} is the right answer to \haslowval{s,q}.
\end{proof}

Complementing the negative result for  $r\in \left(\frac{1}{k}-s,\frac{1}{k}\right)$,
we now present a positive result for $r=1/k$.
Note that we always have $\mms^k_i\leq 1/k$, 
and, moreover,  $\mms^k_i=1/k$ only if there is a partition where each separator has value $0$. 
These observations turn out to be highly useful for the analysis of this case.

\begin{theorem}
\label{thm:pie-1/n-positive}
For pie cutting, there exists an algorithm that, given an agent $i$ and any $k\geq 2$, decides whether $\emph{MMS}^k_i \geq 1/k$ (and if so, computes a maximin partition) using $O(k/s)$ queries in the Robertson--Webb model.
\end{theorem}

\begin{proof}
Since $\mms^k_i\leq 1/k$ always holds, it suffices to decide whether $\mms^k_i = 1/k$.
For the sake of convenience, we modify the queries slightly for pie cutting as follows:
\begin{itemize}
\item $\eval_i(x,y)$: Asks agent $i$ to evaluate the interval $[x,y]$ starting from $x$ and going clockwise to $y$, and return the value $v_i(x,y)$.
\item $\cut_i(x,\alpha)$: For $\alpha\leq 1$, asks agent $i$ to return the first point $y$ going clockwise from~$x$ such that $v_i(x,y)=\alpha$.
\end{itemize}
It is clear that each of these queries can be implemented using no more than two queries in the original model.
If $x > 1$, we identify the point $x$ with the point $x-1$.

\begin{algorithm}
\caption{Determining whether $\mms^k_i = 1/k$ in pie cutting}\label{alg:pie-1/n}
\begin{algorithmic}[1]
\Procedure{MMS-$1/k$-Pie$(k,s)$}{}
\State Divide the pie into $\lceil 2/s\rceil$ equal-length intervals.
\For{each resulting interval $[x,y]$}
\If{$\eval_i(x,y) = 0$}
\State $z\leftarrow \cut_i(x,1)$
\State works $\leftarrow$ True
\For{$j = 1,2,\dots,k$}
\If{$\eval_i(z,z+s) \neq 0$}
\State works $\leftarrow$ False
\EndIf
\State $z\leftarrow \cut_i(z+s,1/k)$
\EndFor
\If{works $=$ True}
\State \Return Yes
\EndIf
\EndIf
\EndFor
\State \Return No
\EndProcedure
\end{algorithmic}
\end{algorithm}

The pseudocode of our algorithm is given as Algorithm~\ref{alg:pie-1/n}.
We first divide the pie into $\lceil 2/s\rceil$ equal-length intervals, so that the length of each interval is at most $s/2$.
For each resulting interval $[x,y]$, if it has value $0$, we find the closest point $z$ going clockwise from~$x$ such that $v_i(x,z) = 1$; note that $z$ is also the farthest point counterclockwise from $x$ such that $v_i(z,x) = 0$.
From point $z$, we check for a potential partition with min-value $1/k$: we insert a separator of size $s$, create a part of value $1/k$, and repeat these steps $k-1$ more times. 
If all $k$ separators have value $0$, return Yes.
Otherwise, if this procedure fails for all candidate intervals $[x,y]$, return No.
Since we have initially divided the pie into $O(1/s)$ intervals and we make $O(k)$ queries for each interval, our algorithm uses $O(k/s)$ queries.

Next, we prove the correctness of the algorithm. 
If the algorithm returns Yes, then the value of each separator is $0$ and the value of each part is $1/k$, so we obtain a $k$-partition of the pie with min-value $1/k$. Hence, in this case $\mms^k_i = 1/k$.

For the converse direction, suppose that $\mms^k_i = 1/k$, so there exists a partition with min-value $1/k$. We can assume without 
loss of generality that each separator
in our partition cannot be extended in either
direction without reducing the value of the adjacent parts, i.e., each separator is an inclusion-maximal interval of value~$0$ (in particular, we allow separators to have length greater than $s$).
Since the length of each separator in this partition is at least $s$, whereas the length of each interval in our initial pie division is at most $s/2$, there exists an interval that is contained in one of the separators; let $[x,y]$ be one such interval.
It suffices to show that the algorithm returns Yes when starting with the interval $[x,y]$ in the for-loop.

Since $[x,y]$ is contained in a separator, we have $v_i(x,y) = 0$.
Denote by $S_1$ the separator containing $[x,y]$.
Let $P_1$ be the adjacent piece of the partition going clockwise, and denote the following separators and pieces by $S_2,P_2,\dots,S_k,P_k$.
Considering all pieces in the clockwise direction, we find that $z$ coincides with the starting point of $S_1$. 
Since $S_1$ has length at least $s$, the separator that the algorithm inserts has value $0$; moreover, when the algorithm jumps by value $1/k$, it will reach exactly the endpoint of~$P_1$, which is also the starting point of $S_2$.
Repeating this argument, we conclude that all $k$ separators that the algorithm inserts indeed have value $0$, and the algorithm returns Yes, as claimed.
\end{proof}

The number of queries made by Algorithm~\ref{alg:pie-1/n} scales linearly with $1/s$. This is in contrast to Theorem~\ref{thm:mms-query-atleast} for cake cutting, where the number of queries is independent of~$s$. We will now show that
for pie cutting the number of queries must depend on $s$; this result holds even for $k=2$
and $r=1/k$.

\begin{theorem}
\label{thm:pie-1/n-negative}
Let $c$ be any constant not depending on $s$.
For pie cutting, there is no algorithm using at most $c$ Robertson--Webb queries that, given an agent $i$, can always decide whether $\emph{MMS}^2_i \geq 1/2$, even when the agent's valuation is piecewise constant (but not given explicitly).
\end{theorem}

\begin{proof}
Assume for contradiction that such an algorithm exists, and let $0 < s < \frac{1}{4c}$.
We will show how an adversary can answer the queries of the algorithm in such a way that after at most $c$ queries, there exists a piecewise constant valuation function consistent with the answers for which $\mms^2_i \geq 1/2$, but also one for which $\mms^2_i < 1/2$.
This is sufficient to obtain the desired contradiction.
Since it always holds that $\mms^2_i \leq 1/2$, the first case is equivalent to $\mms^2_i = 1/2$.

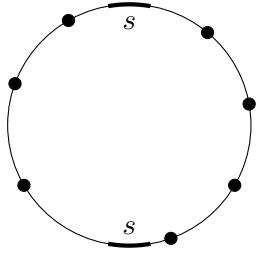
\begin{figure}[!ht]
\centering
\begin{tikzpicture}[scale=0.8]
\draw (3,3) circle [radius = 2];
\draw [ultra thick,domain=80:100] plot ({3+2*cos(\x)}, {3+2*sin(\x)});
\draw [ultra thick,domain=260:280] plot ({3+2*cos(\x)}, {3+2*sin(\x)});
\node at (3,4.7) {$s$};
\node at (3,1.3) {$s$};
\draw [fill] ({3+2*cos(10)},{3+2*sin(10)}) circle [radius = 0.1];
\draw [fill] ({3+2*cos(50)},{3+2*sin(50)}) circle [radius = 0.1];
\draw [fill] ({3+2*cos(120)},{3+2*sin(120)}) circle [radius = 0.1];
\draw [fill] ({3+2*cos(160)},{3+2*sin(160)}) circle [radius = 0.1];
\draw [fill] ({3+2*cos(210)},{3+2*sin(210)}) circle [radius = 0.1];
\draw [fill] ({3+2*cos(290)},{3+2*sin(290)}) circle [radius = 0.1];
\draw [fill] ({3+2*cos(330)},{3+2*sin(330)}) circle [radius = 0.1];
\end{tikzpicture}
\caption{Illustration for the proof of Theorem~\ref{thm:pie-1/n-negative}, with a pie and its recorded points. The two highlighted intervals are ``good'' intervals.}
\label{fig:pie-good-interval}
\end{figure}

The adversary records any point that appears in a query or an answer to it, and answers queries as if the valuation is uniform. 
Suppose that at most $c$ queries have been answered in this manner.
Each query and its answer increase the number of recorded points by at most two, so there are at most $2c$ recorded points.
We say that a closed interval $I$ of length~$s$ is {\em good} if neither $I$ nor the diametrically opposite interval of length~$s$ contains any recorded point 
(see Figure~\ref{fig:pie-good-interval}); a point $x$ is {\em good} if it is the center of a good interval.
For each recorded point, the set of points that it rules out as good is a union of two intervals of length $s$; since $2c\cdot 2s < 1$, a good interval exists. Let $I^+$
be one such interval, and let $I^-$ be the diametrically opposite interval. If necessary, record additional points so that there is at least one recorded point on both pieces of the pie between $I^+$ and $I^-$.
Consider the following two possibilities:

\underline{Possibility 1}: The entire valuation function is uniform. Since $s > 0$, we have $\mms^2_i < 1/2$ in this case.

\underline{Possibility 2}: Consider a piece between any two consecutive recorded points.
If the piece contains neither $I^+$ nor $I^-$, the adversary simply distributes the value  uniformly across the piece.
Else, if the piece contains an interval $I\in\{I^+, I^-\}$, the adversary ``shifts'' the value away from $I$ in the following manner: Divide $I$ into two subintervals of length $s/2$, and for each subinterval, move its value to the adjacent part of the same piece outside $I$.
This can be done so that the resulting valuation is piecewise constant.
Since the two intervals of length $s$ serve as separators for a partition with min-value $1/2$, we have $\mms^2_i = 1/2$.

Therefore, an algorithm using at most $c$ queries cannot distinguish between the case $\mms^2_i<1/2$ and the case $\mms^2_i\ge 1/2$.
\end{proof}

At the other extreme, we show next that deciding whether the maximin share is strictly positive can also be done by a finite algorithm.
Moreover, unlike for the task of deciding whether $\mms^k_i \ge 1/k$, the number of queries needed to decide whether $\mms^k_i > 0$ does not depend on $s$.
For this result we need the assumption $s\le \frac{1}{2k}$, which means that the total length of all separators is at most $1/2$, i.e., the length of a separator does not exceed the average length of the pieces in a partition.
While this is a reasonable assumption since separators are small in most applications, it remains open whether it can be removed.

\begin{theorem}
\label{thm:pie-0}
For pie cutting, there exists an algorithm that, given an agent $i$ and any $k\ge 2$ and $s\le \frac{1}{2k}$, decides whether  $\emph{MMS}^k_i > 0$  using $O(k)$ queries in the Robertson--Webb model.
\end{theorem}

\begin{proof}
We first divide the pie into $2k$ equal-length intervals.
If all of these intervals have strictly positive value, we return Yes.
Else, we consider one of the intervals with value $0$, say $[x,y]$.
We turn the pie into a cake by removing the interval $[x, y]$, 
run the algorithm for deciding whether a cake can be partitioned into $k$ 
parts of value strictly greater than~$r$ (Theorem~\ref{thm:mms-query-greater})
with $r=0$, and report the answer.
It is clear that our algorithm requires $O(k)$ queries.

We will now argue that our algorithm is correct. 
If our algorithm succeeds in the first stage (when it partitions the pie into $2k$ intervals), then since $s\le\frac{1}{2k}$, the odd-numbered intervals serve as separators of a $k$-partition of the pie in which all parts have positive value.
Otherwise, let $I_{xy}$ denote the cake
obtained from the pie by removing the interval $[x, y]$.
If the algorithm
for $I_{xy}$ returns Yes, then the respective partition corresponds 
to a $k$-partition of the pie in which all parts have positive value:
the interval $[x, y]$ provides the missing separator. 

Conversely, suppose
the pie admits a partition $\partition$ into $k$ positive-value parts. 
If our algorithm succeeds in the first stage, we are done, so assume this is not the case. 
We will explain how to transform $\partition$ into a partition $\partition'$
of $I_{xy}$ with $k$ positive-value parts; the existence of $\partition'$
ensures that our algorithm for the cake (cf.~Theorem~\ref{thm:mms-query-greater})
is guaranteed to succeed.

Since all parts of $\partition$ have positive value while the interval $[x,y]$ has value $0$, at most two parts of $\partition$ can overlap 
$[x, y]$. 
If there are exactly two such parts (say, $P_j$ and $P_\ell$, with $x\in P_j$, $y\in P_\ell$), then replacing $P_j$ and $P_\ell$ with $P_j\setminus [x, y]$
and $P_\ell\setminus [x, y]$, respectively, gives rise to a desired partition of $I_{xy}$. 
Else, if exactly one part of $\partition$ overlaps $[x, y]$ (say, $P_j$), then
$P_j\setminus[x, y]$ has at most two connected components, and at least one of these components---say, $P'_j$---has positive value, so we can obtain $\partition'$ from $\partition$ by replacing $P_j$ with $P'_j$.
Finally, if no part of $\partition$ overlaps $[x, y]$, we can simply set $\partition'=\partition$. This completes the proof of correctness.
\end{proof}

We now turn to the problem of computing the maximin share and a maximin partition of a pie. Theorem~\ref{thm:pie-mms-greater} obviously rules out the possibility of exact computation.

\begin{corollary}
\label{cor:pie-mms-compute}
Fix any $k\geq 2$. For pie cutting, there is no finite algorithm in the Robertson--Webb model that, given an agent $i$, can either (a) compute $\emph{MMS}^k_i$, or (b) compute a maximin partition into $k$ pieces for $i$.
This holds even when the valuation of this agent is piecewise constant (but not given explicitly).
\end{corollary}

Despite these negative results, we show that it is possible to approximate the maximin share of an agent up to an arbitrary error.
The idea is to mark points on the pie so that any piece between two adjacent marks has value at most $\varepsilon/2$, and, for each piece between two (not necessarily adjacent) marks, try to construct an $s$-separated partition with min-value equal to the value of this piece, by means of a greedy algorithm.

\begin{theorem}
\label{thm:pie-mms-approx}
Fix any $k\geq 2$. For pie cutting, given an agent $i$ and a number $\varepsilon > 0$, it is possible to find a number $r$ such that $\emph{MMS}^k_i-\varepsilon\leq r\leq \emph{MMS}^k_i$, along with an $s$-separated partition with min-value $r$, using $O(1/\varepsilon)$ queries in the Robertson--Webb model.
\end{theorem}

\begin{proof}
Mark points on the pie so that for any two adjacent marks the value of the piece between them is at most $\varepsilon/2$; let $M$ be the set of marked points.
Let $r$ denote the highest min-value among all $s$-separated partitions such that both endpoints of each of the $k$ pieces are in $M$.
We first claim that $r\geq \mms^k_i-\varepsilon$.
Indeed, consider a maximin partition, and shrink each of the $k$ pieces by moving both endpoints inward until they coincide with marked points.
The resulting partition is still $s$-separated, and the value lost by each piece is at most $\varepsilon/2 + \varepsilon/2 = \varepsilon$.
Since the min-value of the original partition is $\mms^k_i$, the min-value of the new partition is at least $\mms^k_i-\varepsilon$.

Our algorithm starts by constructing the set $M$.
Then, for each interval $[x,y]$ 
where $x, y\in M$,
 we attempt to construct an $s$-separated partition with min-value $v_i(x,y)$ that has $[x, y]$ as one of its pieces and such that the endpoints of all $k$ pieces are in $M$. The construction proceeds in a greedy fashion: starting with $[x, y]$ and going clockwise, we add the smallest separator of length at least $s$ that ends at another marked point (say, $x'$), create the next piece by finding the smallest $y'\in M$ such that 
 $v_i(x', y')\ge v_i(x,y)$, add another separator of length at least $s$ that ends at a marked point, and so on. 
If we can add $k$ 
separators without overlapping with the original interval $[x,y]$, the construction succeeds.
We return the highest value $v_i(x,y)$ such that the construction succeeds, along with the corresponding partition.

Note that the algorithm only needs $O(1/\varepsilon)$ queries in order to mark points in the first step; the additional steps do not require further queries.
To see that the algorithm is correct, consider a partition with min-value $r$ defined in the first paragraph, and let $[x,y]$ be a piece in this partition with value exactly $r$.
The algorithm succeeds when it starts with $[x,y]$: this follows from a greedy argument similar to that in the proof of Theorem~\ref{thm:mms-query-atleast}.
Hence the value returned by the algorithm is at least $r\geq \mms^k_i-\varepsilon$.
Moreover, the algorithm necessarily outputs the min-value of an $s$-separated partition, which is at most $\mms^k_i$ by definition.
This concludes the proof.
\end{proof}

Theorem~\ref{thm:pie-mms-approx} provides us with an arbitrarily close additive approximation of each agent's $1$-out-of-$k$ maximin share; combined with Theorem~\ref{thm:pie-mms-algo}, it enables us to compute an allocation in which each agent~$i$ receives value at least $\mms_i^{n+1}-\varepsilon$ using $O(n^2+n/\varepsilon)$ queries.
The dependence on $\varepsilon$ can be improved by turning the pie into a cake; the argument is similar to the one in the proof of Theorem~\ref{thm:pie-mms-algo}.

\begin{theorem}
\label{thm:pie-mms-approx-improved}
For any pie cutting instance with $n$ agents and any $\varepsilon > 0$, it is possible to compute an allocation in which every agent $i$ receives value at least $\emph{MMS}^{n+1}_i-\varepsilon$ using $O(n^2\log(1/\varepsilon))$ queries in the Robertson--Webb model.
\end{theorem}

\begin{proof}
Cut the pie at point $0$ to turn it into a cake, and remove the interval $[0,s]$.
By Corollary~\ref{cor:mms-allocation-approx}, it is possible to compute an allocation of the remaining cake in which every agent~$i$ receives value at least $\mms^{n,\text{cake}}_i-\varepsilon$ using at most $O(n^2\log(1/\varepsilon))$ queries.
It therefore suffices to show that $\mms^{n,\text{cake}}_i \ge \mms^{n+1,\text{pie}}_i$.
To see this, consider a $1$-out-of-$(n+1)$ maximin partition of the pie, and observe that at most one part overlaps with the interval $[0,s]$.
Therefore, the remaining $n$ parts, along with the $n-1$ separators between them, form an $s$-separated $n$-partition of the cake.
This means that $\mms^{n,\text{cake}}_i \ge \mms^{n+1,\text{pie}}_i$, as desired.
\end{proof}

An immediate corollary of Theorem \ref{thm:pie-mms-approx-improved} 
is that, for each individual agent $i$, it is possible to compute a number $r$ such that 
$\mms^{k+1}_i-\varepsilon\leq r\leq \mms^k_i$
in time $O(k^2\log(1/\varepsilon))$.
However, 
to compute a number $r$ such that 
$\mms^{k}_i-\varepsilon\leq r\leq \mms^k_i$,
we currently need $O(1/\varepsilon)$ queries (Theorem~\ref{thm:pie-mms-approx}).
An interesting question is whether it is possible to make the dependence on $1/\varepsilon$ logarithmic, as is possible for cake cutting (Corollary~\ref{cor:mms-query-bound}).

\begin{open}
In pie cutting, for any $k\ge 2$, is it possible to compute an $\varepsilon$-approximation of $\emph{MMS}^k_i$ using $O(\emph{poly}(k, \log(1/\varepsilon)))$ Robertson--Webb queries?
\end{open}
Notably, a similar exponential gap exists for the $\varepsilon$-approximation of connected envy-free cake-cutting without separation for four or more agents: at least $\Omega(\log(1/\varepsilon))$ queries are required, but the best known algorithm needs $\Theta(1/\varepsilon)$ queries
\citep{branzei2017query}.

Unlike for cake cutting, where there is no hope of giving each agent her $1$-out-of-$n$ maximin share using a finite number of queries (Corollary~\ref{cor:mms-no-partition}), for pie cutting we do not know whether a finite algorithm can ensure every agent her $1$-out-of-$(n+1)$ maximin share.
Nevertheless, we show in Theorem~\ref{thm:mms-2n} that if we further relax our ordinal approximation benchmark, namely, aim to give each agent her $1$-out-of-$2n$ maximin share, finite computation becomes possible. We also prove an analogous result for cake cutting (with $2n$ replaced by $2n-1$) in Theorem~\ref{thm:mms-2n-1}.

\section{Envy-Freeness and Equitability}
\label{sec:EF-EQ}

In this section, we focus on two other well-studied fairness notions: envy-freeness and equitability.

\begin{definition}
An allocation $\allocation = (A_1, \dots, A_n)$ is said to be \emph{envy-free} if $v_i(A_i)\geq v_i(A_j)$ for all $i,j\in N$.
\end{definition}

\begin{definition}
An allocation $\allocation = (A_1, \dots, A_n)$ is said to be \emph{equitable} if $v_i(A_i)= v_j(A_j)$ for all $i,j\in N$.
\end{definition}

An empty allocation is trivially envy-free and equitable, but leaves every agent empty-handed. Thus, envy-freeness and equitability should be combined with other axioms that discourage discarding the entire cake. To this end, 
we will now define a class of allocations (and partitions) that do not waste the cake needlessly.

\begin{definition}
Given a separation parameter $s$, we say that an allocation or partition is \emph{exactly $s$-separated} if any two consecutive pieces are separated by length exactly $s$ (and in case of a cake, the first piece starts at $0$, while the last piece ends at $1$).
\end{definition}

We note that an envy-free or equitable allocation that is exactly $s$-separated may still
leave all agents with zero value: for example, this can happen if all agents have identical valuations and their entire value is packed within an interval of length at most $s$.

\subsection{Existence}

The first question we investigate is whether an exactly $s$-separated envy-free or equitable allocation always exists.
We first consider cake cutting and answer this question in the affirmative for both notions, by adapting the arguments of Simmons \citep{Su99} and \cite{Cheze17} from the setting without the separation requirement.
Since the former argument is well-known in the fair division literature, we only present a sketch here and refer to Section~3 of Su's paper for more details.

\begin{theorem}
\label{thm:envy-free}
For any cake cutting instance, there exists an exactly $s$-separated envy-free allocation.
\end{theorem}

\begin{proof}
Consider all exactly $s$-separated partitions of the cake into $n$ pieces.
The space of all such partitions corresponds to the standard simplex in $\mathbb{R}^n$, with the $i$-th coordinate corresponding to the length of the $i$-th piece; the simplex is scaled down so that its maximum coordinate is $1-(n-1)s$.
In other words, each partition is represented by an $n$-tuple $(b_1,\dots,b_n)$ with $\sum_{i=1}^n b_i = 1 - (n-1)s$, where the $i$-th piece has length~$b_i$.\footnote{In cake cutting without separation \citep[Sec.~3]{Su99}, $s = 0$ and therefore $\sum_{i=1}^n b_i = 1$.} The figure below illustrates, for the case $n=3$ and $s=0.1$, the partition with $(b_1,b_2,b_3) = (0.2,0.25,0.35)$.

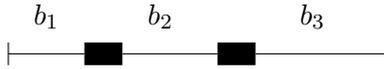
\begin{center}
\begin{tikzpicture}[scale=0.5]
\draw (0,-0.3) -- (0,0.3);

\draw (0,0) -- (2,0);
\draw[line width=3mm] (2,0) -- (3,0);
\node at (1,1) {$b_1$};

\draw (3,0) -- (5.5,0);
\draw[line width=3mm] (5.5,0) -- (6.5,0);
\node at (4,1) {$b_2$};

\draw (6.5,0) -- (10,0);
\node at (8,1) {$b_3$};

\draw (10,-0.3) -- (10,0.3);
\end{tikzpicture}
\end{center}

Consider a triangulation of this simplex by \emph{barycentric subdivision}, and assign each vertex of this triangulation to one of the agents, so that for each small subsimplex it holds that each agent is assigned exactly one vertex of this subsimplex.
(See Section~4 of \citet{Su99}'s paper for details on this step.)
Label each vertex with the index of its assigned agent's favorite piece in the corresponding partition.
If the agent has two or more favorite pieces in the partition, 
then choose one such piece arbitrarily, as long as it is non-empty.
Note that for every vertex its corresponding partition contains at least one non-empty piece, because of the assumption that $s < \frac{1}{n-1}$.
Moreover, since the valuations are monotonic, each agent has at least one non-empty favorite piece. Therefore, this tie-breaker is feasible.
The tie-breaker ensures that the resulting labeling satisfies the conditions of \emph{Sperner's lemma}: each vertex is labeled with an index of a non-zero coordinate.
Therefore, the triangulation has a \emph{Sperner subsimplex}---a subsimplex all of whose
labels are different. Repeating this process with finer triangulations gives an infinite sequence of smaller Sperner subsimplices. This sequence has a subsequence that converges to a single point. By the continuity
of valuations, this limit point corresponds to a partition in
which each agent prefers a different piece, thereby inducing an envy-free allocation.
\end{proof}

Next, we adapt the proof of \citet{Cheze17} to show the existence of an equitable allocation.
Unlike for envy-freeness, for equitability we can additionally choose the order in which the agents are allocated pieces of the cake from left to right.

\begin{theorem}
\label{thm:equitable}
For any cake cutting instance and any ordering of the agents, there exists an equitable
exactly $s$-separated
allocation in which the agents are allocated the pieces from left to right according to the ordering.
\end{theorem}

\begin{proof}
Assume without loss of generality that the desired agent ordering is $1,2,\dots,n$.
Recall that the density function of agent~$i$ is $f_i$.
Consider the sphere 
$$
S^{n-1} := \left\{e=(e_1,\dots,e_n)\in\mathbb{R}^n\,\middle|\, \sum_{i=1}^n e_i^2 = 1 - (n-1)s\right\},
$$ 
and the function $g:S^{n-1}\rightarrow \mathbb{R}^{n-1}$ that maps each point $e\in S^{n-1}$ to 
$$
g(e) = (G_1(e),\dots,G_{n-1}(e)),
$$
where 
$$
G_i(e) = \text{sgn}(e_{i+1})\cdot\int_{e_1^2 + ... + e_i^2 + is}^{e_1^2 + ... + e_i^2 + e_{i+1}^2 + is} f_{i+1}(x)dx - \text{sgn}(e_1)\cdot \int_0^{e_1^2} f_1(x)dx
$$
for $i=1,2,\dots,n-1$.

The function $g$ is continuous, so by the Borsuk--Ulam theorem, there exists $\hat{e}\in S^{n-1}$ such that $g(-\hat{e}) = g(\hat{e})$.
Moreover, one can check that $g(-e)=-g(e)$ for all $e\in S^{n-1}$, and in particular $g(-\hat{e}) = -g(\hat{e})$.
Hence $g(\hat{e})$ is the zero vector, meaning that
$$
\text{sgn}(\hat{e}_{i+1})\cdot\int_{\hat{e}_1^2 + ... + \hat{e}_i^2 + is}^{\hat{e}_1^2 + ... + \hat{e}_i^2 + \hat{e}_{i+1}^2 + is} f_{i+1}(x)dx = \text{sgn}(\hat{e}_1)\cdot \int_0^{\hat{e}_1^2} f_1(x)dx
$$
for every $i=1,2,\dots,n-1$.
Since both integrals are nonnegative, it must be the case that 
$$
\int_{\hat{e}_1^2 + ... + \hat{e}_i^2 + is}^{\hat{e}_1^2 + ... + \hat{e}_i^2 + \hat{e}_{i+1}^2 + is} f_{i+1}(x)dx =  \int_0^{\hat{e}_1^2} f_1(x)dx
$$
for each $i$.
It follows that if we allocate the piece between $\hat{e}_1^2+\dots+\hat{e}_{i-1}^2+(i-1)s$ and $\hat{e}_1^2+\dots+\hat{e}_{i-1}^2+\hat{e}_i^2+(i-1)s$ to agent $i$, we obtain an equitable exactly $s$-separated allocation.
\end{proof}

The existence guarantees carry over to pie cutting. 
Indeed, in order to obtain an envy-free (respectively, equitable) exactly $s$-separated allocation, we can simply insert a separator of length $s$ at an arbitrary position in the pie and apply Theorem~\ref{thm:envy-free} (respectively, Theorem~\ref{thm:equitable}) on the remaining pie (treated as a cake).

We now explore the relationship between envy-freeness/equitability and maximin share fairness. Without separation, it is known that any complete envy-free allocation is proportional, and hence also MMS-fair. 
In Theorem~\ref{thm:ef-implies-mms}, we generalize this observation to the setting with separation.
For this, we need the following lemma.

\begin{lemma}
\label{lem:complete-implies-mms}
(a)
In any exactly $s$-separated partition of a cake into $n$ pieces, each agent~$i$ has value at least $\emph{MMS}_i^{n,s}$ for at least one piece.

(b)
In any exactly $s$-separated partition of a pie into $n$ pieces, each agent~$i$ has value at least $\emph{MMS}_i^{n+1,s}$ for at least one piece.
\end{lemma}
\begin{proof}
(a)
We can represent an exactly $s$-separated partition $\partition=\{P_1, \dots, P_n\}$ 
as a list $(x_0, x_1, \dots, x_n)$ where $x_0=-s$ and $x_n=1$, so that $P_j=[x_{j-1}+s, x_j]$ for each $j\in [n]$.
Now, fix an arbitrary partition $\partition$ represented as $(z_0, z_1, \dots, z_n)$,
and let $(y_0, y_1, \dots, y_n)$ represent some 
exactly $s$-separated maximin partition of agent $i$ 
(such a partition exists by Proposition \ref{prop:alt-def}).
Mark each integer $j\in\{0,\ldots,n\}$ by $L$ if $z_j\leq y_j$ and by $R$ if $z_j\geq y_j$. Note that $0$ and $n$ are marked both $L$ and $R$. 
Therefore, there exists at least one $j\in [n]$ such that $j-1$ is marked $L$ and $j$ is marked $R$.
This means that the piece $[z_{j-1}+s, z_j]$ contains the piece $[y_{j-1}+s, y_j]$, whose value is at least $\mms_i^{n,s}$ since it is a piece in a maximin partition.

(b) Consider an exactly $s$-separated partition $\partition=\{P_1, \dots, P_n\}$ of the pie. We can assume without loss of generality that the leftmost point of $P_1$ is $0$ and hence the rightmost point of $P_n$ is $1-s$. 
Let $\partition'$ be some exactly $s$-separated $1$-out-of-$(n+1)$ maximin partition of the pie.
Remove from the pie the interval $[1-s,1]$, and remove from $\partition'$ the (at most one) part that overlaps $[1-s,1]$; if no such part exists, remove a part closest to $[1-s,1]$. 
Extend the parts adjacent to the removed part so that one of them starts at $0$ and the other one ends at $1-s$. The resulting partition $\partition''$ has $n$ parts, with the first part starting at $0$ and the last part ending at $1-s$, and the value of each part is at least $\textrm{MMS}_i^{n+1,s}$.
Now, the situation (restricted to $[0, 1-s]$) is exactly as in part (a), and the same proof shows that at least one part of $\partition$ contains a part of~$\partition''$.
\end{proof}

The guarantee in part (b) cannot be improved to $\mms_i^{n,s}$.
Indeed, consider Figure~\ref{fig:pie-no-mms} with $s = 1/2-\varepsilon$.
For Alice each of the two parts of length $\varepsilon$ in the left figure has value $1/2$, which means that $\mms_\textrm{Alice}^{n,s} = 1/2$, while each of the two parts of length $\varepsilon$ in the right figure has value $0$ to her.

\begin{theorem}
\label{thm:ef-implies-mms}
In any exactly $s$-separated envy-free allocation, the value of each agent~$i$ is

(a) at least $\emph{MMS}_i^{n,s}$ in case of a cake;

(b) at least $\emph{MMS}_i^{n+1,s}$ in case of a pie.
\end{theorem}
\begin{proof}
In any envy-free allocation, the piece allocated to agent $i$ is at least as valuable to $i$ as each of the $n$ pieces in the allocation. 
Therefore, by Lemma~\ref{lem:complete-implies-mms}, this allocated piece yields value at least $\mms_i^{n,s}$ in case of a cake and $\mms_i^{n+1,s}$ in case of a pie. 
\end{proof}

Similarly to the example following Lemma~\ref{lem:complete-implies-mms}, the bound $\mms_i^{n+1,s}$ in part (b) cannot be improved to $\mms_i^{n,s}$.

\subsection{Computation}

Having established that envy-free/equitable exactly $s$-separated allocations of a cake always exist, it is natural to ask whether they can be computed by a finite algorithm. 
Unfortunately, it turns out that the answer to this question is `no'.
We will show that this is the case even when there are only two agents with identical valuations; note that in this case an exactly $s$-separated allocation
is envy-free if and only if both parts have exactly the same value, 
so in particular envy-freeness is equivalent to equitability.

The argument for the case of cake is simple:
by Theorem~\ref{thm:ef-implies-mms}, if we could compute an envy-free 
(or, equivalently, equitable) allocation in this setting, we could also deduce the maximin share with respect to the common valuation, thereby contradicting Theorem~\ref{thm:mms-impossibility}. 

\begin{corollary}
\label{cor:EF-EQ-impossibility}
For cake cutting, there is no algorithm that can always compute an envy-free or an equitable exactly $s$-separated allocation by asking the agents a finite number of Robertson--Webb queries.
This holds even when $n=2$ and the agents' valuations are identical, piecewise constant, and strictly positive (but not given explicitly).
\end{corollary}

These impossibilities stand in stark contrast to the canonical setting without separation: in that setting, with two agents, an envy-free allocation for non-identical valuations and an equitable allocation for identical valuations can be found by a simple cut-and-choose algorithm.

Regarding pie cutting, 
it is possible to show that no finite algorithm can compute an envy-free/equitable allocation for two identical agents, 
by modifying the proof of Theorem~\ref{thm:mms-impossibility}.\footnote{
We cannot prove the impossibility by reducing from Theorem~\ref{thm:pie-mms-greater},
since an envy-free or equitable division of a pie into $n$ pieces, even with identical valuations, is not necessarily 1-out-of-$n$ MMS-fair.
}
In particular, the adversary can answer the queries made by the algorithm in such a way that after any finite number of queries, the algorithm cannot identify an exactly $s$-separated partition in which the two pieces have the same value based on the given answers.

\begin{theorem}
\label{thm:pie-EF-EQ-impossibility}
For pie cutting, there is no algorithm that can always compute an envy-free or an equitable exactly $s$-separated allocation  by asking the agents a finite number of Robertson--Webb queries.
This holds even when $n=2$ and the agents' valuations are identical, piecewise constant, and strictly positive (but not given explicitly).
\end{theorem}

In light of \Cref{cor:EF-EQ-impossibility} and \Cref{thm:pie-EF-EQ-impossibility}, it would be interesting to develop approximation algorithms that compute allocations with low envy---this direction has been recently pursued in cake cutting without separation \citep{ArunachaleswaranBaKu19,GoldbergHoSu20}.

\section{Conclusion and Future Work}
\label{sec:conclusion}

In this paper, we have initiated the study of cake cutting under separation requirements, which capture scenarios including data erasure in machine processing, cross-fertilization prevention in land allocation, as well as social distancing.
We established several existence and computational results concerning maximin share fairness, both positive and negative.
Overall, our results indicate that maximin share fairness is an appropriate substitute for proportionality in this setting.

Interestingly, several of our positive results, including Theorems~\ref{thm:mms-algo}, \ref{thm:mms-query-atleast}, \ref{thm:mms-query-greater}, and \ref{thm:pie-mms-algo}, do not rely on the assumption that valuations are additive (which is a standard assumption in the cake-cutting literature \citep{Procaccia16}) and hold even for agents with arbitrary monotonic valuations.
This observation reveals another significant advantage of maximin share fairness over proportionality: when valuations are not necessarily additive, even in the absence of separation requirements, no multiplicative approximation of proportionality can be guaranteed.\footnote{To see this, consider agents with identical valuations such that the value of a piece is defined by a non-decreasing function $g(\ell)$ that depends only on the length $\ell$ of the piece. 
Even without separation, in any allocation, at least one of the agents will receive value at most $g(1/n)$, which can be $0$ (or, if $g$ is required to be strictly increasing, arbitrarily close to $0$).
}

We end the paper with a number of directions for future work.

\begin{itemize}
\item Can we improve the query complexity bounds in our results, or establish matching lower bounds?
A particularly interesting question concerns Theorem~\ref{thm:mms-algo}, where we presented an algorithm that computes an MMS-fair allocation using $O(n^2)$ queries given the agents' maximin shares.
Without separation, it is well-known that a proportional allocation can be found using $O(n\log n)$ queries via a divide-and-conquer approach \citep{Even1984Note}, and that this is tight \citep{woeginger2007complexity,EdmondsPr11}.
What is the optimal query complexity of computing an MMS-fair allocation in our setting?

\item While the canonical maximin share is a reasonable fairness requirement when agents have equal entitlements to the resource, in certain situations the agents may be endowed with different entitlements \citep{Barbanel95,cseh2018complexity,ChakrabortyIgSu21,ChakrabortyScSu21}.
Various extensions of the maximin share have been proposed \citep{aziz2019weighted,farhadi2019fair,BabaioffNiTa21,ChakrabortySeSu22}, and it may be interesting to study them in the context of cake cutting with separation. For different entitlements, a connected cake allocation may not exist even without separation
\citep{segalhalevi19cake,crew2020disproportionate}.

\item What happens if we do not require each agent to receive a single connected piece, but instead allow up to $t$ connected pieces for some parameter $t$?
When $t=2$, the analog of Theorem~\ref{thm:mms-algo} no longer holds: one can check that the valuation functions in Figure~\ref{fig:disconnected} do not admit an MMS-fair allocation.
It remains open whether any (ordinal or cardinal) approximation of the maximin share can be attained.

\begin{figure}
\centering
\begin{tikzpicture}[scale=1]
\node at (2,4.5) {$0$};
\node at (6,4.5) {$1/2$};
\node at (10,4.5) {$1$};
\draw[very thick] (2,5) -- (10,5);
\draw (2,4.9) -- (2,5.1);
\draw (3,4.9) -- (3,5.1);
\draw (4,4.9) -- (4,5.1);
\draw (5,4.9) -- (5,5.1);
\draw (6,4.9) -- (6,5.1);
\draw (7,4.9) -- (7,5.1);
\draw (8,4.9) -- (8,5.1);
\draw (9,4.9) -- (9,5.1);
\draw (10,4.9) -- (10,5.1);
\node at (1,5) {Cake};
\draw[thick] (2,6) -- (5,6);
\draw[thick] (7,6) -- (10,6);
\node at (1,6) {Agent~2};
\node at (3.5,6.2) {\small $1/2$};
\node at (8.5,6.2) {\small $1/2$};
\draw[thick] (2,7) -- (3,7);
\draw[thick] (5,7) -- (7,7);
\draw[thick] (9,7) -- (10,7);
\node at (2.5,7.2) {\small $1/4$};
\node at (6,7.2) {\small $1/2$};
\node at (9.5,7.2) {\small $1/4$};
\node at (1,7) {Agent~1};
\end{tikzpicture}
\caption{An example showing that the analog of Theorem~\ref{thm:mms-algo} does not hold when each agent is allowed to receive up to two connected pieces.
The valued intervals of each agent are indicated along with their values, spread uniformly across each interval.
For $s=1/4$, the maximin share of each agent is $1/2$, but no $s$-separated allocation gives both agents a value of at least $1/2$.
}
\label{fig:disconnected}
\end{figure}
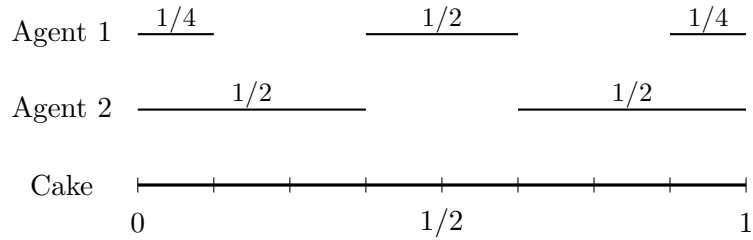

\item  Prior work has explored the combination of fairness and economic efficiency in cake cutting without separation \citep{BeiChHu12,AumannDoHa13,ArunachaleswaranBaKu19}.
With separation, can we achieve maximin share fairness along with certain notions of efficiency, for instance, minimizing the value of the unallocated cake? 
Note that maximin share fairness is always compatible with Pareto efficiency, since a Pareto improvement of an MMS-fair allocation is again MMS-fair.

\end{itemize}

At a higher level, separation requirements represent one type of constraints that arise in a number of applications of cake cutting.
In other applications, it may be desirable to limit the amount of cake that certain agents receive, or ensure that the cake is allocated to agents in a given order.\footnote{\citet{Suksompong21} surveyed different types of constraints in fair division.}
Examining the interplay between such constraints and fairness considerations is an important direction that will likely lead to fruitful research.

\section*{Acknowledgments} 

This work was partially supported by the European Research Council (ERC) under grant number 639945 (ACCORD), by the Israel Science Foundation under grant number 712/20, by the Singapore Ministry of Education under grant number MOE-T2EP20221-0001, and by an NUS Start-up Grant.
Part of the work was done while the third author was a postdoctoral researcher at the University of Oxford.
We would like to thank 
Iosif Pinelis, Fedor Petrov, Jochen Wengenroth, and  Dieter Kadelka
for their mathematical help,
and the anonymous reviewers of the 35th AAAI Conference on Artificial Intelligence (AAAI 2021) and Artificial Intelligence Journal for their valuable comments.

\bibliographystyle{plainnat}
\bibliography{main}

\appendix

\section{Proof of Theorem~\ref{thm:mms-exact}}
\label{app:LP-proof}

Let $v$ be the piecewise constant valuation function of agent $i$, 
described by a list of breakpoints $(p_0, p_1, \dots, p_d)$ with $p_0=0$, $p_d=1$, and a list of densities $(\gamma_1, \dots, \gamma_d)$.
For readability, we set
$I_j :=[p_{j-1}, p_j]$ for each $j\in [d]$, as illustrated below:

\begin{center}
\begin{tikzpicture}[scale=0.7]
\node (p0) at (0,0) {$p_0$};
\node () at (-0.1,.9) {$0$};
\node (p1) at (2,0) {$p_1$};
\node (p2) at (5,0) {$p_2$};
\node (pdots) at (7.5,0) {$\ldots$};
\node (pd) at (10,0) {$p_d$};
\node () at (10.1,.9) {$1$};

\begin{scope}[node distance=0.2,inner sep=1pt]
\node[above = of p0] (p0b) {};
\node[above = of p1] (p1b) {};
\node[above = of p2] (p2b) {};
\node[above = of pd] (pdb) {};
\end{scope}

\node (I1) at (1, 1.4) {$I_1$};
\node (I2) at (3.5, 1.4) {$I_2$};
\node (Idots) at (6.5, 1.4) {$\ldots$};
\node (Id) at (8.5, 1.4) {$I_d$};

\draw (p0b) -- (p1b) -- (p2b) -- (pdb);
\end{tikzpicture}
\end{center}

We also write
\begin{align*}
c^* :=\mms^{n,s}_i, \quad
w_j :=v(0, p_j)=\sum_{\ell=1}^j\gamma_\ell\cdot (p_\ell-p_{\ell-1})
\end{align*}

Since we can check whether $c^*>0$ using Theorem~\ref{thm:mms-query-greater}, we assume from now on that this is the case.
Given two points $y\leq y'$ with $y\in I_\ell$ and $y'\in I_r$,  the value $v(y, y')$ is a linear function of $y$ and $y'$:
\begin{equation}\label{eq:pw-val}
v(y, y') = 
{ \color{blue}(p_{\ell}-y)\gamma_{\ell} }
+
{\color{red} (w_{r-1} - w_{\ell}) }
 +
{\color{ForestGreen} (y'-p_{r-1})\gamma_r }
\end{equation}
This is illustrated below:
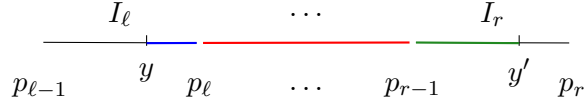
\begin{center}
\begin{tikzpicture}[scale=0.7]
\node (pL1) at (0,0) {$p_{\ell-1}$};
\node (pL)  at (3,0) {$p_{\ell}$};
\node (pdots) at (5,0) {$\ldots$};
\node (pR1) at (7,0) {$p_{r-1}$};
\node (pR)  at (10,0) {$p_{r}$};

\node (IL) at (1.5, 1.4) {$I_{\ell}$};
\node (Idots) at (5, 1.4) {$\ldots$};
\node (IR) at (8.5, 1.4) {$I_r$};

\node (y0) at (2,0.3) {$y$};
\draw (2,0.8) -- (2, 1);

\node (y1) at (9,0.3) {$y'$};
\draw (9,0.8) -- (9, 1);

\begin{scope}[node distance=0.3,inner sep=1pt]
\node[above = of pL1] (pL1b) {};
\node[above = of pL] (pLb) {};
\node[above = of pR1] (pR1b) {};
\node[above = of pR] (pRb) {};
\end{scope}

\draw (pL1b) -- (pLb);
\draw[thick,color=blue] (2,0.85) -- (pLb);
\draw[thick,color=red] (pLb) -- (pR1b);
\draw (pR1b) -- (pRb);
\draw[thick,color=ForestGreen] (pR1b) -- (9,0.85);

\end{tikzpicture}
\end{center}

In what follows, we depart from the notation used in the remainder of the paper and represent an $s$-separated partition of $[0, 1]$ into $n$ parts as $(x_0, x_1, \dots, x_n)$ where $x_0=-s$ and $x_n=1$, so that the $k$-th part of the partition is given by $[x_{k-1}+s, x_k]$.

For each $t\in [0, 1]$, each $k\in [n]$, and each list of intervals 
${\mathcal I}_k = (I_{\ell(1)}, I_{r(1)},\dots,I_{\ell(k)}, I_{r(k)})$
such that 
$\ell(1)\le r(1)\le \dots \le \ell(k)\le r(k)$
and $0\in I_{\ell(1)}$, $t\in I_{r(k)}$, 
consider the following linear program, with variables $x_0, x_1, \dots, x_k, c$:

\begin{framed}
\textbf{Program $\text{LP}_k({\mathcal I}_k, t)$}:
\begin{align}
    \max \quad &c \nonumber\\
    &\text{subject to}\nonumber\\
    &x_0 =-s, \qquad x_k=t\nonumber\\
    &x_{q-1}+s \in I_{\ell(q)}, \quad x_q\in I_{r(q)}\qquad\text{for all $q\in [k]$}
    \label{eq:lp-pw-intervals}\\
    &x_{q-1}+s \le x_{q} \qquad\text{for all $q\in [k]$}\nonumber\\
    &v(x_{q-1}+s, x_q) \ge c \qquad\text{for all $q\in [k]$}
    \label{eq:lp-pw-val}
\end{align}
\end{framed}
Note that $\text{LP}_k({\mathcal I}_k, t)$ is indeed a linear program; in particular, constraints~\eqref{eq:lp-pw-intervals} are linear, because the endpoints of each $I_j$ are known, and constraints~\eqref{eq:lp-pw-val} are
linear because, by \eqref{eq:pw-val}, the function $v(y, y')$ is a linear function of $y$ and $y'$
as long as the intervals containing $y$ and $y'$ are known.
Every feasible solution of this linear program corresponds to an $s$-separated partition of $[0, t]$ into $k$ parts, 
where the beginning and the end of the $q$-th part (i.e., $x_{q-1}+s$ and $x_q$ in the linear program), $1\le q\le k$, are contained in the intervals $I_{\ell(q)}$ and $I_{r(q)}$, respectively;
thus, the solution of this linear program is the maximum value among all such partitions. It follows that we can compute $\mms_i$ by solving a linear program as long as we know the `correct' intervals for both endpoints of each part of an $s$-separated maximin partition. We now show how to identify such intervals in polynomial time.

Given $k\in [n]$, a list of intervals ${\mathcal I}_k=(I_{\ell(q)}, I_{r(q)})_{q\in [k]}$, and a partition $(x_0, \dots, x_n)$ of 
$[0, 1]$,\footnote{Recall that, as indicated earlier, in this proof we use $(x_0,\dots,x_n)$ to represent an $s$-separated partition of $[0,1]$ into $n$ parts.} we will say that 
$(x_0, \dots, x_n)$ and ${\mathcal I}_k$ are {\em consistent}
if $x_{q-1}+s \in I_{\ell(q)}$ and
$x_q\in I_{r(q)}$ for each $q\in [k]$.
If $(x_0, \dots, x_n)$ is an $s$-separated maximin partition of $[0, 1]$,
and ${\mathcal I}_n=(I_{\ell(q)}, I_{r(q)})_{q\in [n]}$ is consistent 
with $(x_0, \dots, x_n)$,
then the solution of the linear program $\text{LP}_{n}({\mathcal I}_n, 1)$
is $c^*$.

To build a list ${\mathcal I}_n$
that is consistent with some $s$-separated maximin partition,
we proceed inductively. First, 
we construct a list ${\mathcal I}_1=(I_{\ell(1)}, I_{r(1)})$
that is consistent with some $s$-separated maximin partition.
Then, for each $1<k\le n$
we extend the list ${\mathcal I}_{k-1}=(I_{\ell(q)}, I_{r(q)})_{q\in [k-1]}$
to a list ${\mathcal I}_{k}=(I_{\ell(q)}, I_{r(q)})_{q\in [k]}$ so that
${\mathcal I}_k$ is also consistent with some $s$-separated maximin partition.

\smallskip

\noindent{\bf Base case\ } We compute ${\mathcal I}_1=(I_{\ell(1)}, I_{r(1)})$ as follows.
Since the first part of any partition starts at $0$, we set $\ell(1)=1$.
To compute $r(1)$, 
for each $j\in \{0,\ldots,d\}$,
set $w_j:=v(0, p_j)$, and use the algorithm from Theorem~\ref{thm:mms-query-atleast} to check\footnote{Even though the algorithm in Theorem~\ref{thm:mms-query-atleast} is designed for $\mms^{n,s}_i(0,1)$, it can be easily adapted to our task by allowing $n-1$ pieces in the partition and considering the interval $[p_j+s,1]$ instead of $[0,1]$.
Similar statements hold for other applications of the algorithm in the remainder of this proof.} whether $w_j \leq \mms^{n-1,s}_i(p_j+s, 1)$.
Pick the first $j$ for which the answer is No, and set $I_{r(1)}=I_j$, as illustrated below:

\begin{center}
\begin{tikzpicture}[scale=0.7]
\node () at (-0.1,.9) {$0$};
\node (p0) at (0,0) {$p_0$};
\node (p1) at (3,0) {$p_1$};
\node (I1) at (1.5, 1.4)  {\small $I_1 = I_{\ell(1)}$};

\node (pdots1) at (4.5,0) {$\ldots$};
\node (Idots) at (4.5, 1.4) {$\ldots$};

\node (pj1) at (6,0) {$p_{j-1}$};
\node (pj) at (9,0) {$p_{j}$};
\node (pjs) at (9.5,0.5) {\small $s$};
\node (Ij) at (7.5, 1.4) {\small $I_{j} = I_{r(1)}$};
\node (MMSn) at (12, 1.4) {\small $[\mms_i^{n-1,s}<w_j]$};

\node (pdots2) at (13.5,0) {$\ldots$};
\node (pd) at (15,0) {$p_d$};
\node () at (15.1,.9) {$1$};

\begin{scope}[node distance=0.2,inner sep=1pt]
\node[above = of p0] (p0b) {};
\node[above = of p1] (p1b) {};
\node[above = of pj1] (pj1b) {};
\node[above = of pj] (pjb) {};
\node[above = of pd] (pdb) {};
\end{scope}

\draw (p0b) -- (p1b) -- (pj1b) -- (pjb) -- (pdb);
\draw (10,0.6) -- (10,0.85);
\end{tikzpicture}
\end{center}

To see that this approach is correct, consider two $s$-separated maximin partitions 
of $[0, 1]$, which we denote by $(x^-_0, \dots, x^-_n)$ and $(x^+_0, \dots, x^+_n)$: these partitions are chosen so that for every $s$-separated maximin partition $(x_0, \dots, x_n)$ of $[0, 1]$
we have $x^-_1\le x_1\le x^+_1$. 
(These two partitions may coincide.)
Since $c^*>0$, we have
$0<x^-_1$ and $x^+_1<1$.
Suppose that $p_{j^--1}< x^-_1\le p_{j^-}$ and $p_{j^+-1}\le x^+_1 < p_{j^+}$
for some $j^-, j^+\in [d]$.
For every $z\in [x^-_1, x^+_1]$, the partition
$(x^-_0, z, x^+_2, \dots, x^+_n)$ is also an
$s$-separated maximin partition of $[0, 1]$.
Thus, for every $j$ such that $j^-\le j\le j^+$, there exists an $s$-separated
maximin partition of $[0, 1]$ such that the right endpoint of the first part lies in $I_j$, 
i.e., any such choice of $j$ is suitable for $r(1)$; we will argue that our algorithm selects 
$r(1)$ so that $j^-\le r(1)\le j^+$.

\begin{center}
\begin{tikzpicture}[scale=0.7]
\node (pL1) at (0,0) {$p_{j^- -1}$};
\node (pL)  at (3,0) {$p_{j^-}$};
\node (pdots) at (5,0) {$\ldots$};
\node (pR1) at (7,0) {$p_{j^+ -1}$};
\node (pR)  at (10,0) {$p_{j^+}$};

\node (IL) at (1.5, 1.4) {$I_{j^-}$};
\node (Idots) at (5, 1.4) {$\ldots$};
\node (IR) at (8.5, 1.4) {$I_j^+$};

\begin{scope}[node distance=0.3,inner sep=1pt]
\node[above = of pL1] (pL1b) {};
\node[above = of pL] (pLb) {};
\node[above = of pR1] (pR1b) {};
\node[above = of pR] (pRb) {};
\end{scope}

\draw (pL1b) -- (pLb) -- (pR1b) -- (pRb);

\node (y0) at (2,0.3) {$x_1^-$};
\draw (2,0.8) -- (2, 1);

\node (y1) at (9,0.3) {$x_1^+$};
\draw (9,0.8) -- (9, 1);

\node (z) at (4,0.25) {$z$};
\draw (4,0.8) -- (4, 1);

\end{tikzpicture}
\end{center}

Indeed, by our choice of $x^-_1$ we have $v(0, x^-_1)=c^*$: if $v(0, x^-_1)<c^*$, then
the value of the first part is less than $c^*$, and if $v(0, x^-_1)>c^*$, we can find 
an $x<x^-_1$ with $v(0, x)=c^*$, a contradiction with our choice of $x^-_1$. 
By the same argument, $v(0, x) < v(0, x^-_1)=c^*$ for every $x<x^-_1$
and hence $v(0, p_{j'})< c^*$ for every $j'<j^-$. On the other hand, for $j'<j^-$, $[p_{j'}+s, 1]$ is a superset of $[x^-_1+s, 1]$ and hence $\mms_i^{n-1,s}(p_{j'}+s, 1) \ge c^*$. 
Thus, $r(1)\ge j^-$. 
By a similar argument, $v(0, p_{j^+})\ge v(0, x^+_1)\ge c^*$.
Further, if $\mms_i^{n-1,s}(p_{j^+}+s, 1) \ge c^*$, then, by combining the corresponding maximin partition with $[0, p_{j^+}]$,
we obtain an $s$-separated partition of $[0, 1]$ into $n$ parts such that the value of each part is at least $c^*$ and the first part ends at $p_{j^+}>x^+_1$, a contradiction with our choice of $x^+_1$.
Thus, no such partition of $[p_{j^+}, 1]$ exists and hence our algorithm selects $r(1)\le j^+$.
This concludes the proof for $k=1$.

\smallskip

\noindent{\bf Inductive step\ }
Now, fix an integer $k$ with $1<k\le n$, and suppose we have already constructed a list 
${\mathcal I}_{k-1} = (I_{\ell(q)}, I_{r(q)})_{q\in [k-1]}$ that is consistent with some $s$-separated
maximin partition of $[0, 1]$. Our goal is to compute $I_{\ell(k)}, I_{r(k)}$
so that the list ${\mathcal I}_{k} = (I_{\ell(q)}, I_{r(q)})_{q\in [k]}$ is consistent with some $s$-separated maximin partition of $[0, 1]$.

We start with $I_{\ell(k)}$. 
Among all $s$-separated maximin partitions that are consistent with ${\mathcal I}_{k-1}$, consider a partition with the smallest
value of $x_{k-1}$ and one with the largest value of $x_{k-1}$; we denote these partitions by 
$(x^-_0, \dots, x^-_n)$ and $(x^+_0, \dots, x^+_n)$, respectively.
Since $c^*>0$, we have
$x^+_{k-1}+s<1$.
Suppose that 
$p_{j^--1}< x^-_{k-1}+s \le p_{j^-}$ and $p_{j^+-1}\le x^+_{k-1}+s < p_{j^+}$
for some $j^-, j^+\in [d]$.
Note that for every $z$ such that $x^-_{k-1} \le z \le x^+_{k-1}$ it holds that
$(x^-_0 \dots, x^-_{k-2}, z, x^+_{k}, \dots, x^+_n)$ is an $s$-separated 
maximin partition. 
Therefore, for every $j$ such that $j^-\le j\le j^+$,
there exists an $s$-separated maximin partition $(y_0, \dots, y_n)$ 
of $[0, 1]$ that is consistent with ${\mathcal I}_{k-1}$ and satisfies $y_{k-1}+s\in I_j$.
Hence, any $j$ in this range is a suitable choice for $\ell(k)$.

\begin{center}
\begin{tikzpicture}[scale=0.7]
\node (pL1) at (-2,0) {$p_{j^- -1}$};
\node (pL)  at (2,0) {$p_{j^-}$};
\node (pdots) at (5,0) {$\ldots$};
\node (pR1) at (7,0) {$p_{j^+ -1}$};
\node (pR)  at (12,0) {$p_{j^+}$};

\node (IL) at (1.5, 1.4) {$I_{j^-}$};
\node (Idots) at (5, 1.4) {$\ldots$};
\node (IR) at (8.5, 1.4) {$I_j^+$};

\begin{scope}[node distance=0.3,inner sep=1pt]
\node[above = of pL1] (pL1b) {};
\node[above = of pL] (pLb) {};
\node[above = of pR1] (pR1b) {};
\node[above = of pR] (pRb) {};
\end{scope}

\draw (pL1b) -- (pLb) -- (pR1b) -- (pRb);

\node (y0) at (0,0.3) {$x_{k-1}^- + s$};
\draw (0,0.8) -- (0, 1);

\node (y1) at (10,0.3) {$x_{k-1}^+ + s$};
\draw (10,0.8) -- (10, 1);

\end{tikzpicture}
\end{center}

Let $\mathcal L$ be the collection of all 
intervals $I_j$ that have a non-empty intersection with $[p_{r(k-1)-1}+s, p_{r(k-1)}+s]$,
where $r(k-1)$ is the index of the rightmost interval in ${\mathcal I}_{k-1}$.
For every interval $I_j\in \mathcal L$, 
we solve $\text{LP}_{k-1}({\mathcal I}_{k-1}, p_j-s)$; let $c'$ be the solution of this linear program.
We then check whether $\mms_i^{n-k+1,s}(p_j, 1) \ge c'$ using the algorithm from Theorem~\ref{thm:mms-query-atleast}, and set $\ell(k)$ to be the smallest $j$ for which this is not the case.

To see the correctness of our approach, first observe that it suffices to restrict our attention to the intervals in $\mathcal L$.
Indeed, for every maximin partition $(x_0, \dots, x_n)$
that is consistent with ${\mathcal I}_{k-1}$ we have $x_{k-1}\in I_{r(k-1)}$, and hence $p_{r(k-1)-1}+s \le x_{k-1}+s\le p_{r(k-1)}+s$. Thus, $x_{k-1}+s$
has to be contained in an interval $I_j$ such that 
$I_j\cap [p_{r(k-1)-1}+s,  p_{r(k-1)}+s]\neq\emptyset$.

Next, we will argue that (a)  $\ell(k) \ge j^-$, and (b) $\ell(k) \le j^+$. 
To prove (a), consider some $j'<j^-$. We claim that 
the solution of $\text{LP}_{k-1}({\mathcal I}_{k-1}, p_{j'}-s)$ is strictly less than $c^*$.
Indeed, otherwise there is an $s$-separated partition of $[0, p_{j'}-s]$ into $k-1$
parts that is consistent with ${\mathcal I}_{k-1}$ and such that the value of each part is at least $c^*$; since $p_{j'} \leq p_{j^--1} < x_{k-1}^-+s$, by combining this partition with $[p_{j'}, x^-_k],[x^-_k+s, x^-_{k+1}], \dots, [x^-_{n-1}+s, 1]$,
we obtain a maximin partition of $[0, 1]$ that is consistent with ${\mathcal I}_{k-1}$,
a contradiction with our choice of $j^-$. Also, we claim that $\mms_i^{n-k+1,s}(p_{j'}, 1) \ge c^*$: again, this is witnessed
by $[p_{j'}, x^-_k],[x^-_k+s, x^-_{k+1}], \dots, [x^-_{n-1}+s, 1]$. Hence, our algorithm will not set $\ell(k)=j'$.

To prove (b), observe that since $p_{j^+}-s > x^+_{k-1}$, we have that $(x^+_0, x^+_1, \dots, x^+_{k-2}, p_{j^+}-s, c^*)$ is a feasible solution for $\text{LP}_{k-1}({\mathcal I}_{k-1}, p_{j^+}-s)$, i.e., the solution of this LP is at least $c^*$. Now, if $\mms_i^{n-k+1,s}(p_{j^+}, 1) \ge c^*$, we can combine the corresponding partition with $[0, x^+_1],[x^+_1+s, x^+_2], \dots$$,$ $[x^+_{k-2}+s, p_{j^+}-s]$
to obtain an $s$-separated maximin partition of $[0, 1]$ where the $k$-th part starts at $p_{j^+}$,
a contradiction with our choice of $j^+$. Thus, we have $\ell(k)\le j^+$.

Combining (a) and (b), we conclude that $j^-\leq \ell(k)\leq j^+$, so our choice of $\ell(k)$ works.

To compute $I_{r(k)}$, we proceed in a similar fashion. Specifically, let $\mathcal R$ be the collection of intervals $(I_j)_{\ell(k)\le j\le d}$.
For each interval $I_j$ in this collection, 
let ${\mathcal I}_k(j)$ be the list obtained by taking ${\mathcal I}_{k-1}$ 
and appending $I_{\ell(k)}$ (which we have computed in the previous step) and $I_j$ to it.
We then solve $\text{LP}_k({\mathcal I}_k(j), p_j)$; let $c'$ be the solution of this linear program.
We check whether $\mms_i^{n-k,s}(p_j+s, 1) \ge c'$, pick the smallest $j$ for which this is not the case, and set $I_{r(k)}=I_j$. 

The proof of correctness for this case is similar. We know that the set of maximin partitions $(x_0, \dots, x_n)$ that are consistent with ${\mathcal I}_{k-1}$ and satisfy $x_{k-1}+s\in I_{\ell(k)}$ is not empty. Among all such partitions, consider one with the minimum $x_k$ and one with the maximum $x_k$. Let $I_{j^-}$ and $I_{j^+}$ be the respective intervals containing the $k$-th cut $x_k$;
we claim that $j^-\le r(k)\le j^+$. Both inequalities are proved in exactly the same way as for $\ell(k)$.
This completes the inductive argument.

Both for $I_{\ell(k)}$ and for $I_{r(k)}$, we can speed up the search for the appropriate $j$ by using binary search on ${\mathcal L}$ (respectively, $\mathcal R$) rather than checking each interval in that list; with this modification, we can identify each interval in ${\mathcal I}_n$ by solving $O(\log n)$ linear programs.

To see that our algorithm runs in polynomial time, it remains to observe that, to find the maximin share of agent $i$, we first need to identify all intervals in ${\mathcal I}_n$, and then solve the resulting linear program. Thus, we need to solve $O(n\log n)$ linear programs, and the size of each linear program is polynomial in the size of the input.
The remaining steps of the algorithm (such as invoking the algorithm from Theorem~\ref{thm:mms-query-atleast}) can also be implemented efficiently.

\section{Ordinal Maximin Relaxations}
\label{app:ordinal-cake}

We prove here that in cake cutting, it is possible to compute an allocation in which each agent receives her $1$-out-of-$(2n-1)$ maximin share using a finite number of queries.
Note that this result is incomparable to Corollary~\ref{cor:mms-allocation-approx}, which shows that an arbitrarily close additive approximation of the $1$-out-of-$n$ maximin share can be computed.

\begin{theorem}
\label{thm:mms-2n-1}
For any cake cutting instance with $n$ agents, it is possible to compute an allocation in which every agent $i$ receives value at least $\emph{MMS}^{2n-1}_i$ using $O(n^2/s)$ queries in the Robertson--Webb model.
\end{theorem}

\begin{proof}
From the left end of the cake, we repeatedly move a knife to the right by length~$s$.
After each move, we ask each agent $i$ whether the piece to the left of the knife has value at least $\mms^{2n-1}_i$. To implement this query, we first ask the agent to evaluate the piece~$P$ to the left of the knife so as to obtain $r=v_i(P)$, run the algorithm that can decide whether $\mms^{2n-1}_i > r$ (see Theorem~\ref{thm:mms-query-greater}), and flip the answer.
If the answer is Yes for at least one agent, we allocate the piece to one such agent; we then remove this agent along with the adjacent piece of length~$s$, and recurse on the remaining agents and cake.
If there is only one agent left, that agent receives all of the remaining cake.
Asking an agent can be implemented using at most $1 + (2(2n-1)-1) = 4n-2$ queries (Theorem~\ref{thm:mms-query-greater}).
Since we move the knife $O(1/s)$ times, each time asking at most $n$ agents, we need $O(n^2/s)$ queries.

We now prove the correctness of the algorithm. 
Consider any agent $i$ and her $1$-out-of-$(2n-1)$ maximin partition $\partition = \{P_1, \dots, P_{2n-1}\}$; assume that $P_j=[x_j, y_j]$ for $j\in [2n-1]$.
Let $[z_j, t_j]$ be the $j$-th piece allocated by the algorithm, where
$z_1=0$ and $z_{j+1}=t_j+s$; for notational convenience, let $t_0=-s$. 
If $i$ receives a piece during the first $n-1$ steps of the algorithm, she values this piece at least $\mms^{2n-1}_i$. To complete the proof, we will argue by induction on $j$ that if agent 
$i$ does not receive any of the first $j$ pieces allocated by the algorithm, 
then the remaining cake, i.e., $[t_j+s, 1]$, 
contains the piece $P_{2j+1}$ of her partition.
This implies both that the algorithm will be able to allocate a piece to each agent and that the last agent values her piece at least $\mms^{2n-1}_i$.

For $j=0$ our claim is trivially true. Now, suppose it has been established for $j'<j$, and agent $i$ did not receive any of the first $j$ pieces.
We know that $[t_{j-1}+s, 1]$ contains $P_{2j-1}$. Consider the piece $[z_j, t_j] = [t_{j-1}+s, t_j]$. Either this piece is of length $s$ or agent~$i$ did not say Yes when the knife was at $t_j-s$. 
Either way, $t_j-s$ is no further to the right than the right endpoint of $P_{2j-1}$, i.e., $y_{2j-1}$.
That is, $t_j - s \le y_{2j-1}$.
As we have $x_{2j}\ge y_{2j-1}+s$, this implies $t_j\le x_{2j}$, and hence $t_j+s\le x_{2j}+s\le x_{2j+1}$. That is, $[t_j+s, 1]$ contains the piece $P_{2j+1}$, as claimed. This completes the induction and hence establishes the correctness of our algorithm.
\end{proof}

If we want to cut a pie, we can turn it into a cake by cutting it at an arbitrary point (e.g., at point $0$) and removing an interval of length $s$ starting from this point, similarly to the proof of Theorem~\ref{thm:pie-mms-approx-improved}. 

\begin{theorem}
\label{thm:mms-2n}
For any pie cutting instance with $n$ agents, it is possible to compute an allocation in which every agent $i$ receives value at least $\emph{MMS}^{2n}_i$ using $O(n^2/s)$ queries in the Robertson--Webb model.
\end{theorem}

\section{$\cutr$ Query}
\label{sec:cutright}

We show that the $\cutr_i(x,\alpha)$ query, which returns the rightmost point $y$ for which $v(x,y) = \alpha$, cannot be implemented using finitely many queries in the standard Robertson--Webb model.
This query has been used by \citet{CechlarovaPi12}, where it was called a ``reverse cut'' query.

\begin{theorem}
For any agent $i$ and real number $\alpha\in(0,1)$, 
let $\cutr_i(0,\alpha)$ be the largest $x\in(0,1)$ for which $v_i(0,x)=\alpha$.
There is no algorithm that computes $\cutr_i(0,\alpha)$
by asking agent $i$ a finite number of Robertson--Webb queries.
\end{theorem}

\begin{proof}
Suppose for contradiction that such an algorithm exists.
During the run of the algorithm, there is always a finite set of points $x\in[0,1]$ for which the algorithm knows the value of $v_i(0,x)$; we say that such points are \emph{recorded}. 
Initially, only points $0$ and~$1$ are recorded.
An adversary can answer all queries as if $v_i$ is uniform, i.e., for any two consecutive recorded points on the cake, if the piece between them has length $t$, then it also has value $t$.
After any finite number of queries, it is possible that the entire valuation is uniform, in which case the algorithm should answer $\alpha$.
But it is also possible that, for some small $\varepsilon>0$, it holds that $v_i(\alpha,\alpha+\varepsilon)=0$, where $\alpha+\varepsilon$ is smaller than the smallest recorded point that is strictly larger than $\alpha$. In this case, the algorithm should answer at least $\alpha+\varepsilon$.
\end{proof}

\section{Generalized Ordinal Maximin Guarantees}
\label{app:l-out-of-k}

Given a pie and two parameters $\ell < k$, an agent's \emph{$\ell$-out-of-$k$ maximin share} is defined as the maximum value that the agent can secure for herself by choosing an $s$-separated partition of the pie into $k$ pieces and getting the $\ell$ lowest-value pieces.\footnote{This  definition is an adaptation of a similar definition by \citet{BabaioffNiTa21}, which was formulated for indivisible goods.}

We will now show that Theorem~\ref{thm:pie-mms-algo} can be generalized to guarantee to each agent her $\ell$-out-of-$(\ell n+1)$ maximin share, for any integer $\ell \geq 1$; Theorem~\ref{thm:pie-mms-algo} corresponds to the special case $\ell = 1$.
As an example where this can be useful, suppose $s=1/6$ and $n=2$, and consider an agent who values the regions $[0, 1/30]$, $[6/30, 7/30]$, $[12/30, 13/30]$, $[18/30, 19/30]$, $[24/30, 25/30]$ uniformly at $1/5$ each (and has no value for the remaining pie).
Then her $2$-out-of-$5$ maximin share is $2/5$, which is higher than her $1$-out-of-$3$ maximin share. In contrast, if she values the regions $[0, 1/6]$, $[2/6, 3/6]$, $[4/6, 5/6]$ uniformly at $1/3$ each, then her $1$-out-of-$3$ maximin share is $1/3$, which is higher than her $2$-out-of-$5$ maximin share.

In fact, we can even generalize Theorem~\ref{thm:pie-mms-algo} in a \emph{pluralistic} manner, by showing that for each agent~$i$ and each $\ell_i\in\mathbb N$ 
we can guarantee to agent $i$ her $\ell_i$-out-of-$(\ell_i n+1)$ maximin share.

\begin{theorem}
\label{thm:pie-l-out-of-k}
For any pie cutting instance with $n$ agents, and any positive integers $\ell_1,\dots,\ell_n$, there exists an allocation in which every agent~$i$ receives a piece of value at least $\emph{MMS}^{\ell_i\emph{-out-of-}(\ell_in+1)}_i$.
Moreover, given the value of $\emph{MMS}^{\ell_i\emph{-out-of-}(\ell_in+1)}_i$ for each agent~$i$, such an allocation can be computed using $O(n^2)$ queries in the Robertson--Webb model.
\end{theorem}

\begin{proof}
We ask each agent $i$ to mark the leftmost point $x_i$ (i.e., the first such point when moving clockwise from $0$) such that $v_i(0,x_i) = \mms^{\ell_i\text{-out-of-}(\ell_in+1)}_i$.
The agent who marks the leftmost $x_i$ is allocated the piece $[0,x_i]$ (with ties broken arbitrarily); we then remove this agent along with the piece $[x_i,x_i+s]$, and recurse on the remaining agents and pie.
If there is only one agent left, we still allocate to that agent a piece worth $\mms^{\ell_i\text{-out-of-}(\ell_in+1)}_i$ (rather than the entire remaining pie).
Since we make $n-j$ \textsc{Cut} queries when there are $n-j$ agents left (and no \textsc{Eval} queries), our algorithm uses $\sum_{j=0}^{n-1}(n-j) = O(n^2)$ queries.

We now prove the correctness of the algorithm.
Consider any agent $i$ and her $\ell_i$-out-of-$(\ell_in+1)$ maximin partition $\partition=\{P_1, \dots, P_{\ell_in+1}\}$.
For $j\in \{1,2,\dots,\ell_in\}$, let $Q_j=P_{j+1}$ if $0$~is in the interior of $P_1$, and $Q_j=P_j$ otherwise; we
will write $Q_j=[x_j, y_j]$.
Note that the segment $[0, y_j]$ contains $j$ parts of $\partition$.
If agent $i$ receives the first piece allocated by the algorithm, she receives value $\mms^{\ell_i\text{-out-of-}(\ell_in+1)}_i$.
Else, the right endpoint of the allocated piece is no further to the right than $y_{\ell_i}$.
Since the algorithm inserts a separator of length exactly~$s$, the right endpoint of the first separator is no further to the right than $x_{\ell_i + 1}$.
Applying a similar argument repeatedly, we find that if agent $i$ is not allocated any of the first $n-1$ pieces by the algorithm, then after removing the $(n-1)$-st piece and the following separator, the remaining cake contains $Q_{\ell_i(n-1) + 1},\dots,Q_{\ell_in}$. 
Now, if $0$ is in the interior of $P_1$, then the remaining cake also contains a positive amount of $P_1$ as well as the separator between $P_{\ell_in+1}=Q_{\ell_in}$ and $P_1$. 
On the other hand, if $0$ is not in the interior of $P_1$, then $Q_{\ell_i n}=P_{\ell_i n}$ and the remaining cake contains $P_{\ell_in+1}$ as well as the separator between $P_{\ell_in}$ and $P_{\ell_in+1}$. 
In either case, the remaining cake contains $Q_{\ell_i(n-1) + 1},\dots,Q_{\ell_in}$ as well as the separator that comes after $Q_{\ell_in}$.
Hence, if we allocate a piece of value $\mms^{\ell_i\text{-out-of-}(\ell_in+1)}_i$ to $i$, its right endpoint is no further to the right than $y_{\ell_in}$ and thus the remaining cake contains an unallocated segment of length $s$, which will serve as a separator between the piece that was allocated first and the piece that was allocated last. 
It follows that in either case the resulting allocation is $s$-separated.
\end{proof}

\end{document}